\documentclass[journal]{IEEEtran}
\usepackage{amssymb,amsmath,amsthm,enumitem}

\usepackage{graphicx}
\usepackage{caption}
\usepackage{subcaption}
\usepackage{algorithm}
\usepackage{algorithmicx}
\usepackage{booktabs}
\usepackage[labelformat=parens,labelsep=quad,skip=3pt]{caption}

\ifCLASSINFOpdf
\else
\fi
\usepackage{amsmath}
\usepackage{amssymb}  
\usepackage{amsfonts}
\usepackage{mathrsfs}
\usepackage{mathptmx} 
\usepackage{mathtools}

\usepackage{graphicx}
\usepackage{textcomp}
\usepackage{xcolor}
\usepackage{subcaption}
\usepackage{placeins}

\usepackage{algorithm}
\usepackage{algorithmicx}
\usepackage{algpseudocode}

\usepackage{multirow}
\usepackage{subcaption}
\usepackage{csquotes}

\usepackage[hidelinks]{hyperref}

\usepackage{cite}

\newtheorem{assumption}{Assumption}
\newtheorem{lemma}{Lemma}
\newtheorem{remark}{Remark}
\newtheorem{theorem}{Theorem}

\newtheorem{definition}{Definition}
\newtheorem{problem}{Problem}

\begin{document}
%
\title{SAFE--MA--RRT: Multi-Agent Motion Planning with Data-Driven Safety Certificates}

%

\author{Babak~Esmaeili
        and~Hamidreza~Modares*
\thanks{
This work is supported by the National Science Foundation under award ECCS-2227311. \\
B. Esmaeili and H. Modares are with the Department of Mechanical Engineering, Michigan State University, East Lansing, MI, 48863, USA. (e-mails: esmaeil1@msu.edu and modaresh@msu.edu). }}


%
%
\maketitle

\begin{abstract}
This paper proposes a fully data-driven motion-planning framework for homogeneous linear multi-agent systems that operate in shared, obstacle-filled workspaces without access to explicit system models. Each agent independently learns its closed-loop behavior from experimental data by solving convex semidefinite programs that generate locally invariant ellipsoids and corresponding state-feedback gains. These ellipsoids, centered along grid-based waypoints, certify the dynamic feasibility of short-range transitions and define safe regions of operation. A sampling-based planner constructs a tree of such waypoints, where transitions are allowed only when adjacent ellipsoids overlap, ensuring invariant-to-invariant transitions and continuous safety. All agents expand their trees simultaneously and are coordinated through a space–time reservation table that guarantees inter-agent safety by preventing simultaneous occupancy and head-on collisions. Each successful edge in the tree is equipped with its own local controller, enabling execution without re-solving optimization problems at runtime. The resulting trajectories are not only dynamically feasible but also provably safe with respect to both environmental constraints and inter-agent collisions. Simulation results demonstrate the effectiveness of the approach in synthesizing synchronized, safe trajectories for multiple agents under shared dynamics and constraints, using only data and convex optimization tools.
\end{abstract}


\IEEEpeerreviewmaketitle

\section{Introduction}

\IEEEPARstart{M}{}ultiagent systems (MAS) are used in drone swarms, delivery fleets, warehouse robots, and traffic networks. They promise scalability and robustness by distributing work across many agents in the same workspace. To make this practical, motion planning must produce coordinated, safe, and dynamically feasible trajectories for all agents, while enforcing state limits, avoiding obstacles, and preventing inter-agent collisions on a shared time scale.

Motion planning for autonomous systems has been studied widely, with well-known families including graph search \cite{hart1968astar}, potential fields \cite{khatib1986potential}, optimization-based trajectory methods\cite{zucker2013chomp,schulman2014trajopt}, and sampling-based planners such as probabilistic roadmaps (PRM) \cite{kavraki1996probabilistic} and rapidly-exploring random trees (RRT) \cite{wang2020neural,jiang2021r2}. These methods work well for single agents, but extending them to multiple agents raises a tight coupling in space and time. Agents must reason about each other’s future occupancy to avoid conflicts, often under dynamics and constraints that simple geometric path planners do not capture. Multi-agent path finding (MAPF) addresses collisions on grids using space–time reservations and conflict resolution \cite{silver2005cooperative,sharon2015cbs,boyarski2015icbs,standley2010finding}, while velocity-obstacle methods provide decentralized collision avoidance in continuous spaces \cite{vandenberg2008rvo,vandenberg2011orca}. Still, many multi-agent methods assume kinematic models, decouple planning from feedback, or do not certify that the planned motion is feasible under real dynamics and state bounds.

Invariant–set–based planners certify safety by connecting nodes only when the current node lies inside the next node’s invariant region, so that controller switching preserves safety \cite{weiss2017saferrt,danielson2017path,danielson2020robust,niknejad2024soda}; however, these methods are model-based, typically assuming known linear dynamics to compute ellipsoidal (or polytopic) sets via LMIs. This assumption is limiting when dynamics are uncertain. While set-theoretic control formalizes safety via $\lambda$-contractivity \cite{blanchini1999invariance}, global invariance over complex admissible regions is rarely achievable; the largest invariant subset depends on controller structure and data richness \cite{bisoffi2020data,depersis2021formulas}. Partitioned designs with local controllers can help when a single linear feedback is insufficient \cite{nguyen2023,nguyen2024}, yet they remain model-based, and even though computing invariant ellipsoids is tractable for linear systems, the model dependence reduces practicality when a reliable model is not available \cite{greiff2024invariant}. This motivates data-driven approaches—indirect (identify then control) or direct (synthesize controllers directly from data) \cite{hou2013model,wang2016indirect,bisoffi2022controller}—that enable model-free controller synthesis for LTI systems via convex programs and avoid identification error compounding \cite{modares2023data,depersis2021formulas,vanwaarde2020fundamental}. Recent work has also investigated data-based invariant/safety certificates \cite{greiff2024invariant,bisoffi2020data,depersis2021formulas}. Nevertheless, while single-agent invariant–set–based data-driven motion planners have begun to appear—e.g., a data-driven planner for uncertain systems \cite{esmaeili2025data} and performance-aware planner \cite{niknejad2025dasp}—to the best of our knowledge there is still no multi-agent counterpart. In multi-agent settings, sampling-based planners such as PRM/RRT may still propose waypoints without verifying (purely from data) safe navigability, which motivates integrating data-driven safety guarantees with space–time coordination \cite{karaman2011rrtstar,kavraki1996probabilistic,schouwenaars2001mixed}.

This paper presents \emph{SAFE--MA--RRT}, a 
\textbf{S}afe \textbf{A}nd \textbf{F}easible \textbf{E}llipsoidal framework for \textbf{M}ulti-\textbf{A}gent motion planning via  \textbf{R}apidly-exploring \textbf{R}andom \textbf{T}rees. SAFE--MA--RRT is a data-driven, dynamics-aware motion-planning framework for linear multi-agent systems with discrete-time dynamics. Using one informative trajectory per agent, we compute local state-feedback gains and contractive ellipsoidal invariant sets by solving semidefinite programs that depend only on collected data; these certificates are embedded in a grid-based RRT, where a one-cell move is accepted only if an SDP certifies an invariant ellipsoid inside the local polytope spanning the two cells. To coordinate agents on a shared time grid, we expand all trees synchronously and employ a space–time reservation table to prevent same-cell conflicts and head-on swaps \cite{silver2005cooperative}. Although our coordination stage is inspired by synchronous expansion and reservation tables \cite{silver2005cooperative}, these methods cannot be seamlessly applied to a data-driven invariant-set planner: they do not specify how to identify and manage active invariant ellipses during search, nor how to enforce $\lambda$-contractive switching purely from data. We therefore develop an alternative that (i) ties each reserved space–time cell to its certified child ellipsoid, so the reservation table explicitly tracks the active ellipses that govern execution; (ii) admits a move only when parent/child ellipsoids overlap, ensuring safe controller switching; and (iii) adjusts untrackable nominal nodes to nearby points inside the child ellipse before commit, preserving dynamics awareness without additional model assumptions.

The result is a planner that grows all agent trees in synchronized layers and enforces safety at two levels: workspace safety via invariant-set containment in local polytopes and inter-agent safety via space–time reservations. In contrast to geometric MAPF, every accepted edge carries its own feedback law and invariant set, ensuring dynamic feasibility and constraint admissibility at execution time; in contrast to model-based invariant-set planners, our gains and certificates are learned directly from data \cite{depersis2021formulas,vanwaarde2020fundamental}. Simulations with single- and two-agent scenarios show that the method certifies each motion step with data-driven ellipsoids and produces coordinated paths that obey state limits. The framework offers a practical bridge between sampling-based planning and data-driven control, and sets the stage for probabilistic safety under noise and for extensions to heterogeneous and nonlinear multi-agent systems.

\noindent\textbf{Notation:}
We denote by $\mathbb{R}$ the set of real numbers, by $\mathbb{R}^n$ the $n$-dimensional real vector space, and by $\mathbb{N}_0=\{0,1,2,\dots\}$ the set of nonnegative integers. The $n\times n$ identity matrix is denoted by $I_n$, and $\mathbf{0}_n$ denotes the $n\times n$ zero matrix.
For a matrix $A$, $A_i$ is its $i$-th row and $A_{ij}$ the $(i,j)$ entry. If $A$ and $B$ have the same size, the symbol
$A(\le,\ge)B$ is to be understood elementwise, i.e.,
$A_{ij}(\le,\ge)B_{ij}$ for all $i,j$.
For any matrix $A$, $A^\top$ denotes its transpose. In block-symmetric expressions, the symbol $(\ast)$ indicates the
transpose block needed to complete symmetry.
For a symmetric matrix $Q$, the relations $Q\succeq 0$ and $Q\preceq 0$ mean that $Q$ is positive semidefinite and negative semidefinite, respectively. Given a set $\mathcal S$ and a scalar $\mu\ge 0$, the scaled set
$\mu\mathcal S:=\{\mu x:\,x\in\mathcal S\}$.

A directed graph is a pair $G=(V,E)$ with a finite vertex set $V$ and a
set of ordered pairs $E\subseteq V\times V$ called directed edges. An
edge $(u,v)\in E$ points from tail $u$ to head $v$ and encodes the
permitted direction of travel. A (directed) path is a sequence of
vertices $(v_0,v_1,\dots,v_\ell)$ such that $(v_{k-1},v_k)\in E$ for
all $k=1,\dots,\ell$. Graph-search procedures seek a path between
designated vertices while satisfying any imposed constraints.


\smallskip
The following definitions are used to describe admissible and safe sets in this paper.


\begin{definition}\label{def_2}
A polytope is the intersection of a finite number of half-spaces and is expressed as 
\begin{equation}\label{eq.polyhedral_set}
    \mathcal{S} = \{ x \in \mathbb{R}^n \mid F x \leq g \},
\end{equation}
where $F \in \mathbb{R}^{q \times n}$ and $g \in \mathbb{R}^q$ define the constraints of the polytope.
\end{definition}

\begin{definition}\label{def_3}
An ellipsoidal set, denoted by $\mathcal{E}(P, c)$, is defined as  
\begin{equation}
    \mathcal{E}(P, c) := \{ x \in \mathbb{R}^n \mid (x - c)^\top P^{-1} (x - c) \leq 1 \},
\end{equation}
where $P \in \mathbb{R}^{n \times n}$ is a symmetric positive definite matrix, and $c \in \mathbb{R}^n$ represents the center of the ellipsoid.
\end{definition}

\begin{definition}\label{definition_6} 
(\textbf{Admissible Set}): An admissible set represents the collection of states that satisfy the system's physical and operational constraints. \end{definition}

\begin{definition}\label{definition_7} 
(\textbf{Safe Set}): A safe set is a subset of the admissible set that remains invariant under the system dynamics. If the system starts within the safe set, it will stay within the set for all future time steps under certain conditions. 
\end{definition}

\begin{definition}[{\(\lambda\)-Contractive and Positive Invariant Sets}]\label{definition_8}
For the system~\eqref{eq:lti_agent_i}, let $\lambda \in (0,1]$.  
An ellipsoidal set $\mathcal{E}(P,c)$ is called $\lambda$-contractive if, for every $x(k) \in \mathcal{E}(P,c)$, it holds that
\[
x(k+1) \in \lambda \mathcal{E}(P,c).
\]

In the special case $\lambda = 1$, the set $\mathcal{E}(P,c)$ is called positively invariant.
\end{definition}

\section{Preliminaries and Problem Formulation}
The multi–agent system consists of \(N_a\) homogeneous agents, each described by the same linear time‑invariant (LTI) dynamics.  
For the \(i\)-th agent, \(i\in\{1,\dots,N_a\}\), the discrete‑time state‑space model is  
\begin{align}\label{eq:lti_agent_i}
x_i(k+1) &= A\,x_i(k) + B\,u_i(k), \\
y_i(k)   &= C\,x_i(k),
\end{align}
where \(x_i(k)\in\mathbb{R}^{n}\) and \(u_i(k)\in\mathbb{R}^{m}\) denote, respectively, the state and control input of agent \(i\) at the sampling instant \(k\).  
Because the agents are homogeneous, the pair \((A,B)\) is identical for every agent.  
The output matrix \(C\) is chosen such that \(y_i(k)\in\mathbb{R}^2\) represents the planar position of agent \(i\).

\begin{assumption}\label{assump:unknown_matrices}
The exact system matrices $(A,B,C)$ are unknown. However, the full state vector $x(k)$ is available for measurement.
\end{assumption}

\begin{assumption}\label{assump:controllability}
The actual system defined by $(A,B)$ is controllable.
\end{assumption}

\begin{assumption}\label{ass:Y_union}
Consider a team of $N_a$ homogeneous agents.
For every agent index $i\in \{1,\dots,N_a\}$ and discrete time $k\in\mathbb N_0$,
the position output $y_i(k)\in\mathbb R^{2}$ is required to satisfy
\begin{equation}
        y_i(k)\in\mathcal Y ,
\end{equation}
where the admissible position set $\mathcal Y\subseteq\mathbb R^{2}$ 
is (possibly) non-convex but admits a finite union–of–polytopes representation
\begin{equation}
        \mathcal Y
        \;=\;
        \bigcup_{\kappa\in\mathcal I_{\mathcal Y}}
        \mathcal Y_{\kappa},
        \qquad
        |\mathcal I_{\mathcal Y}|<\infty .
\end{equation}

Each polytope $\mathcal Y_{\kappa}\subseteq\mathbb R^{2}$ is compact and described by linear inequalities
\begin{equation}
        \mathcal Y_{\kappa}
        =
        \bigl\{\,y\in\mathbb R^{2}
        \;\bigl|\;
        F_{\mathcal Y_\kappa}(l)^{\top}y\le g_{\mathcal Y_\kappa}(l),
        \;l=1,\dots,n_f^{\kappa}\bigr\}.
\end{equation}
\end{assumption}

Each agent must move from its prescribed start position
$y_{i,\mathrm s}\in\mathcal Y$ to a goal position
$y_{i,\mathrm g}\in\mathcal Y$ while satisfying:

\noindent\textit{(i) workspace safety:}
$y_{i}(k)\in\mathcal Y$ for all $k\in\mathbb N_0$ and all $i$;

\noindent\textit{(ii) pairwise separation:}
for all $i \neq j$ and all $k \in \mathbb{N}_0$, 
the trajectories of agents $i$ and $j$ remain in their respective 
invariant safe sets, and these sets do not overlap.

\noindent\textit{(iii) goal reachability:}
there exists $\varepsilon>0$ such that
$\displaystyle\limsup_{k\to\infty}
\|y_{i}(k)-y_{i,\mathrm g}\|_2 \le \varepsilon$.

\subsection{Data‑Driven Representation for Each Agent}\label{sec:data_driven}
To eliminate any dependence on the unknown model matrices \(A\) and \(B\),
we represent the closed-loop behaviour of every agent directly from a
short, informative input–state trajectory.  The basic idea is simple:
excite the system with a finite sequence of inputs, log the
corresponding states, and use those data alone to reconstruct the map
that propagates the state one step ahead.  Once this map is available,
we can synthesise a feedback gain without ever identifying
\(A\) or \(B\), thereby making the subsequent safety-certificate
construction entirely model-free and scalable to large multi-agent
teams.

For the multi-agent setting, consider each agent $i \in \{1,\dots,N_a\}$ applying an input sequence
$[u_i(0),\,u_i(1),\,\ldots,\,u_i(N-1)]$
to its own system. During this process, $N$ consecutive state samples are collected and arranged into data matrices:
\begin{align}
& U_{0,i} = [u_i(0),\,u_i(1),\,\ldots,\,u_i(N{-}1)], \label{eq:MAS_Data_U0} \\
& X_{0,i} = [x_i(0),\,x_i(1),\,\ldots,\,x_i(N{-}1)], \label{eq:MAS_Data_X0} \\
& X_{1,i} = [x_i(1),\,x_i(2),\,\ldots,\,x_i(N)], \label{eq:MAS_Data_X1} \\
& Y_{0,i} = [y_i(0),\,y_i(1),\,\ldots,\,y_i(N-1)].
\label{eq:MAS_Data_Y0}
\end{align}

These matrices are then used to construct the local invariant-set certificates and feedback gains for each agent in the simultaneous RRT--Ellipse procedure.

\begin{assumption}[Persistently‑Exciting Data]\label{stack_full_rank}
For every agent $i \in N_a$, the stacked matrix
\begin{equation}\label{eq.rank}
\begin{bmatrix} U_{0,i}\\[2pt] X_{0,i}\end{bmatrix}\in\mathbb R^{(m+n)\times N}
\end{equation}
has full row rank and $N_i\ge m+n$.
\end{assumption}


\begin{lemma}\label{lem:G_factorisation}
Let Assumption 3 hold.  
Then, for each agent $i$ there exist matrices
$K_i\in\mathbb R^{m\times n}$ and
$G_i\in\mathbb R^{N\times(m+n)}$ such that
\begin{equation}\label{eq:KG_factorisation}
\begin{bmatrix}
      K_i & I_m\\[2pt] I_n & 0
\end{bmatrix}
\;=\;
\begin{bmatrix}
      U_{0,i}\\[2pt] X_{0,i}
\end{bmatrix}G_i.
\end{equation}
Partition $G_i$ as $G_i=[\,G_{1,i}\;\;G_{2,i}\,]$ with
$G_{1,i}\in\mathbb R^{N\times n}$ and $G_{2,i}\in\mathbb R^{N_i\times m}$.
Under the state-feedback law
\begin{equation}\label{eq:data_gain}
      u_i(k)=K_i\,x_i(k),
\end{equation}
the closed-loop dynamics and input/output matrices satisfy
\begin{align}
    x_i(k+1) &= X_{1,i}\,G_{1,i}\,x_i(k), \label{eq:data_cl_dynamics_x}\\
    B &= X_{1,i}\,G_{2,i}, \label{eq:data_cl_dynamics_B}\\
    C &= Y_{0,i}\,G_{1,i}. \label{eq:data_cl_dynamics_C}
\end{align}
\end{lemma}


\begin{proof}
Since the stacked data matrix in Assumption~\ref{stack_full_rank}
is full row rank, there exists a right inverse
$G_i \in \mathbb R^{N\times(m+n)}$ such that
\begin{equation}\label{eq:Gi}
\begin{bmatrix}
  U_{0,i}\\[2pt] X_{0,i}
\end{bmatrix} G_i
=
\begin{bmatrix}
  K_i & I_m\\[2pt] I_n & 0
\end{bmatrix}.
\end{equation}

Partition $G_i$ as $G_i = [\,G_{1,i}\;\;G_{2,i}\,]$ with
$G_{1,i}\in\mathbb R^{N\times n}$ and
$G_{2,i}\in\mathbb R^{N\times m}$.

From the system dynamics \eqref{eq:lti_agent_i} and collected data, we have
\begin{equation}\label{eq:X1}
X_{1,i} = A X_{0,i} + B U_{0,i}.
\end{equation}

We multiply both sides by $G_{1,i}$.  
Using~\eqref{eq:Gi}, we know that
\[
X_{0,i} G_{1,i} = I_n, \qquad
U_{0,i} G_{1,i} = K_i.
\]

Therefore,
\begin{equation}
X_{1,i} G_{1,i} = AX_{0,i} G_{1,i} + B U_{0,i} G_{1,i}
= A + B K_i.
\end{equation}

This gives the closed-loop dynamics
\begin{equation}
x_i(k+1) = (A+BK_i)x_i(k) = X_{1,i} G_{1,i} x_i(k),
\end{equation}
which proves the equality in \eqref{eq:data_cl_dynamics_x}.

Multiplying~\eqref{eq:X1} by $G_{2,i}$ and using~\eqref{eq:Gi}, we have
\[
X_{0,i} G_{2,i} = 0, \qquad
U_{0,i} G_{2,i} = I_m,
\]
so
\begin{equation}
X_{1,i} G_{2,i} = AX_{0,i} G_{2,i} + B U_{0,i} G_{2,i}
= B.
\end{equation}

Thus, the equality in \eqref{eq:data_cl_dynamics_B} follows.

By definition of the measured data, we have
\begin{equation}
Y_{0,i} = C X_{0,i}.
\end{equation}

Multiplying both sides by $G_{1,i}$ and recalling that
$X_{0,i} G_{1,i} = I_n$, it follows that
\begin{equation}
Y_{0,i} G_{1,i} = C X_{0,i} G_{1,i} = C.
\end{equation}

This proves the equality in \eqref{eq:data_cl_dynamics_C}.
\end{proof}



\subsection{Feedforward Control Design}
Fix a constant output reference $r_i \in \mathbb R^{2}$.  
To track $r_i$, we first compute a steady-state pair
$(\bar x_i,\bar u_i)$ satisfying
\begin{align}\label{eq:steady_state}
\bar x_i &= A \bar x_i + B \bar u_i, \\
r_i &= C \bar x_i. \label{eq.r_i}
\end{align}

This ensures that when $x_i(k)=\bar x_i$ and $u_i(k)=\bar u_i$, the
system remains at equilibrium with output $y_i(k)=r_i$.

By defining $T:=\begin{bmatrix}
A-I & B \\[3pt]
C & 0
\end{bmatrix}$, equations \eqref{eq:steady_state} and \eqref{eq.r_i} can be written compactly as
\begin{equation}\label{eq:steady_matrix}
T \begin{bmatrix}
\bar x_i \\[3pt] \bar u_i
\end{bmatrix}
=
\begin{bmatrix}
0 \\[3pt] r_i
\end{bmatrix},
\end{equation}
so that the steady-state pair $(\bar x_i,\bar u_i)$ is obtained by
solving the linear system
\begin{equation}\label{eq.steady_state}
\begin{bmatrix}
\bar x_i \\ \bar u_i
\end{bmatrix}
=
T^{\!\!-1}
\begin{bmatrix}
0 \\ r_i
\end{bmatrix}.
\end{equation}

\begin{assumption}\label{assumption_10}
The matrix $T$ is invertible.
\end{assumption}

Since $(A,B,C)$ are not known explicitly, they are replaced with their
data-driven counterparts given in Lemma~\ref{lem:G_factorisation}.
Thus, $(\bar x_i,\bar u_i)$ can be computed directly from measured
trajectories without requiring a model of the plant as follows
\begin{equation}\label{eq.equilibrium_calculation}
\begin{bmatrix}
\bar x_i \\ \bar u_i
\end{bmatrix}
=
\left(
\underbrace{\begin{bmatrix}
X_{1,i}(I-G_{2,i}U_{0,i})G_{1,i} & X_{1,i}G_{2,i} \\[3pt]
Y_{0,i}G_{1,i} & 0
\end{bmatrix}}_{\hat T_i}
\right)^{-1}
\begin{bmatrix}
0 \\ r_i
\end{bmatrix}.
\end{equation}

\begin{assumption}\label{assumption_11}
The matrix $\hat{T}_i$ is invertible for $i\in\{1,\dots,N_a\}$.
\end{assumption}

\subsection{Error-Dynamics Formulation (per agent)}
Defining the error state
\begin{equation}\label{eq:error_state}
e_i(k) := x_i(k) - \bar x_i,
\end{equation}
and applying the
affine feedback law
\begin{equation}\label{eq.controller}
u_i(k)=K_i\bigl(x_i(k)-\bar x_i\bigr)+\bar u_i,
\end{equation}
the closed-loop error dynamics follow as
\begin{equation}\label{eq:error_dynamics_i}
e_i(k+1) = X_{1,i}G_{1,i}e_i(k),
\end{equation}
which evolve directly under the data-driven closed-loop matrix $X_{1,i}G_{1,i}$.

\subsection{Problem Formulation}
The following problem statement summarizes the safe multi-agent motion-planning and control task studied in this work.

\begin{problem}[Data-Driven Multi-Agent Motion Planning]\label{prob:multi_agent_motion}
Consider the LTI dynamics~\eqref{eq:lti_agent_i}, the admissible output set $\mathcal{Y}$ and agent-specific start/goal positions $(y_{i,\mathrm s}, y_{i,\mathrm g})$.  
For each agent $i$, let $\mathcal{P}^{\mathrm{cand}}_i = \{q_{i,j}\}_{j=1}^{N_c^{(i)}} \subset \mathcal{Y}$ denote a finite set of candidate waypoints.  

The objective is to design:
\begin{itemize}
    \item a feasible sequence of waypoints 
    $\{p_{i,\ell}\}_{\ell=0}^{N_w^{(i)}} \subseteq \mathcal{P}^{\mathrm{cand}}_i$,
    with $p_{i,0} = y_{i,\mathrm s}$ and $p_{i,N_w^{(i)}}$ lying within an $\varepsilon$-neighborhood of $y_{i,\mathrm g}$;
    \item a corresponding sequence of steady-state pairs $(\bar{x}_{i,\ell}, \bar{u}_{i,\ell})$ and low-level affine state-feedback controllers
    \begin{equation}\label{eq.controller_final}
        u_i(k) = K_{i,\ell}\bigl(x_i(k) - \bar{x}_{i,\ell}\bigr) + \bar{u}_{i,\ell},
    \end{equation}
    valid whenever the segment $(p_{i,\ell}\!\rightarrow p_{i,\ell+1})$ is active, ensuring safe transitions between successive waypoints.
\end{itemize}
The design must guarantee that requirements \textit{(i)}–\textit{(iii)} are satisfied for every agent $i$ and all $k \in \mathbb{N}_0$.
\end{problem}

The steps to address Problem~\ref{prob:multi_agent_motion} are outlined as follows:

\begin{enumerate}
    \item Discretize the admissible output space $\mathcal{Y}$ into a grid.  
    Select a set of candidate waypoints $\mathcal{P}^{\mathrm{cand}}_i = \{q_{i,j}\}_{j=1}^{N_c^{(i)}}$
    from this grid and compute the corresponding equilibrium pairs $(\bar{x}_{i,j}, \bar{u}_{i,j})$ using the steady-state equation~\eqref{eq.steady_state}.  

    \item For each candidate waypoint $q_{i,j}$, solve a data-driven feedback design problem to obtain a stabilizing gain $K_{i,j}$ together with an invariant ellipsoidal safe set $\mathcal{E}_{i,j}$ that incorporates safety constraints.  

    \item Check connectivity: if $\mathcal{E}_{i,j}$ overlaps with at least one previously accepted invariant set of the same agent and does not intersect the active ellipsoids of other agents (ensuring inter-agent separation), then store $(\bar{x}_{i,j}, \bar{u}_{i,j}, K_{i,j}, \mathcal{E}_{i,j})$ as a node in the waypoint graph $\mathcal{G}$. Otherwise, discard the candidate waypoint.  

    \item Repeat steps 1–3 until a feasible sequence of selected waypoints 
    $
       \{p_{i,\ell}\}_{\ell=0}^{N_w^{(i)}} \subseteq \mathcal{P}^{\mathrm{cand}}_i$
    is obtained, where the associated invariant ellipsoids overlap and connect the start point $y_{i,\mathrm{s}}$ to an $\varepsilon$-neighborhood of the goal $y_{i,\mathrm{g}}$.  
\end{enumerate}

In the next section, we present a data-driven method to synthesize the required state-feedback gains.


\section{Data-Driven Safe Feedback Control Synthesis}\label{sec:inv_set}
The objective of this section is to compute, for each agent~$i$, an affine state-feedback controller of the form~\eqref{eq.controller} such that the closed-loop system states remain within a prescribed admissible region.  

Without loss of generality, we first consider a regulation problem with the equilibrium pair $(\bar x_i,\bar u_i)=(0,0)$. In this case, the controller reduces to a static gain $K_i \in \mathbb{R}^{m\times n}$, i.e., $u_i(k) = K_i x_i(k)$. The design task is to synthesize a gain $K_i$ that ensures safety. To this end, the gain $K_i$ is parameterized by means of a $\lambda$-contractive ellipsoidal safe set $\mathcal{E}_i(P_i,0)$, which is constrained to lie within a polyhedral admissible set
\begin{equation}
    \mathcal{P}^i = \{\,e \mid F_i\,e \le g_i\,\}.
\end{equation}

The same design procedure is applied for each waypoint index $\ell$, 
yielding gains $K_{i,\ell}$ associated with the corresponding 
ellipsoidal safe set $\mathcal{E}_{i,\ell}$. By working in error coordinates $e_i(k)$, the admissible set 
$\mathcal{P}_i$ is centered at the origin by construction.

We now formalize the control objective in the following problem statement.

\begin{problem}[Data-Driven Safe Control for Agent $i$]\label{prob:agent_safe}
Consider the dynamics of agent $i$ described in \eqref{eq:lti_agent_i} under Assumptions~1--4 and 6.  
The admissible set of agent $i$ is denoted by $\mathcal{P}_i(F_i,g_i)$, and we represent the safe set by the ellipsoid $\mathcal{E}_i(P_i,0)$. The objective is to determine a data-driven state-feedback gain $K_i$ such that the ellipsoid $\mathcal{E}_i(P_i,0)$ is invariant and $\lambda$-contractive within the admissible set $\mathcal{P}_i$.
\end{problem}

\begin{theorem}[Safe Controller Design for Each Agent]\label{thm:sdp_invariance_MAS}
Let agent \(i\in\{1,\dots,N_a\}\) satisfy Assumptions 1–-4 and 6 and let the data
matrices \((U_{0,i},X_{0,i},X_{1,i},Y_{0,i})\) be organized as in
\eqref{eq:MAS_Data_U0}–-\eqref{eq:MAS_Data_Y0}. If there exist matrices
\(P_i\in\mathbb S^{n}\), \(S_i\succeq0\), and \(G_{2,i}\) that solve
\begin{subequations}\label{eq:opt_MAS}
\begin{align}
\max_{P_i,S_i}\;& -\log\det(P_i) \tag{\ref{eq:opt_MAS}a}\label{eq:opt_obj_MAS}\\
\mathrm{s.t.}\;&
\begin{bmatrix}
      P_i & X_{1,i}S_i \\[2pt]
      *   & \lambda P_i
\end{bmatrix}\succeq0,\tag{\ref{eq:opt_MAS}b}\label{eq:opt_con1_MAS}\\
&\begin{bmatrix}
      P_i & P_iF_{r}^{i^{\!\top}} \\[2pt]
      *   & g_{r}^{i^2}
\end{bmatrix}\succeq0,\quad r=1,\dots,q,\tag{\ref{eq:opt_MAS}c}\label{eq:opt_con2_MAS}\\
&X_{0,i}S_i=P_i,\tag{\ref{eq:opt_MAS}d}\label{eq:opt_con3_MAS}\\
&X_{0,i}G_{2,i}=0,\;U_{0,i}G_{2,i}=I_m,\tag{\ref{eq:opt_MAS}e}\label{eq:opt_con4_MAS}
\end{align}
\end{subequations}
then Problem 2 is solved, and the ellipsoid $\mathcal E_i(P_i,0):=\bigl\{e\in\mathbb R^{n}\,\bigl|\,e^{\!\top}P_i^{-1}e
   \le1\bigr\}$
is the largest contractive subset of the admissible polytope
\(\mathcal P_i\) for the error dynamics in
Eq.\,\eqref{eq:error_dynamics_i}.  
Moreover, the feedback gain is
\begin{equation}\label{eq:gain_MAS_final}
   K_i = U_{0,i}S_iP_i^{-1}.
\end{equation}
\end{theorem}

\begin{proof}
For the error system $e_i(k+1)=X_{1,i}G_{1,i} e_i(k)$, $\lambda$‑contractivity is expressed as
\begin{equation}
    e_i(k+1)^{\!\top}P_i^{-1}e_i(k+1)
    \;\le\;
    \lambda\,e_i(k)^{\!\top}P_i^{-1}e_i(k),
    \qquad 0<\lambda<1.
    \label{eq:lambda_contract_A}
\end{equation}

Substituting \eqref{eq:error_dynamics_i} into
\eqref{eq:lambda_contract_A} yields the quadratic form
\begin{equation}
    e_i(k)^{\!\top}
    \bigl(G_{1,i}^{\!\top}X_{1,i}^{\!\top}P_i^{-1}
          X_{1,i}G_{1,i}-\lambda P_i^{-1}\bigr)
    e_i(k)\le 0, \label{eq:bmi_A}
\end{equation}
which is bilinear in the decision variables $(G_{1,i},P_i)$.  
To turn \eqref{eq:bmi_A} into a linear matrix inequality (LMI) we
define the lifting variable
\begin{equation}
    S_i := G_{1,i}P_i,          \label{eq:lift_A}
\end{equation}
so that $G_{1,i}=S_iP_i^{-1}$.  Substituting \eqref{eq:lift_A} into
\eqref{eq:bmi_A} and applying the Schur complement produces the LMI
constraint \eqref{eq:opt_con1_MAS}.

The ellipsoid
$\mathcal E_i(P_i,0)=\{e\mid e^{\!\top}P_i^{-1}e\le 1\}$
is required to lie inside the polytope
$\mathcal P_i=\{e\mid F_i\,e\le g_i\}$.  
For each facet index $r=1,\dots,q$ this is equivalent to the inequality
\begin{equation}
    \max_{e\in\mathcal E_i(P_i,0)}F_{r}^i e \;\le\; g_{r}^i,    \label{eq:max_facet_A}
\end{equation}
which, by the Cauchy–Schwarz argument
$ee^{\!\top}\preceq P_i$ for $e\in\mathcal E_i(P_i,0)$, is equivalent
to
\begin{equation}
    F_{r}^i P_iF_{r}^{{i}^{\!\top}}\;\le\; g_{r}^{i^{2}}. \label{eq:facet_quad_A}
\end{equation}

Applying the Schur complement to \eqref{eq:facet_quad_A} yields the
facet LMIs \eqref{eq:opt_con2_MAS}.  
The equalities
\begin{equation}
    X_{0,i}Y_i=P_i,\qquad
    X_{0,i}G_{2,i}=0,\qquad
    U_{0,i}G_{2,i}=I_m,          \label{eq:data_consistency_A}
\end{equation}
which stem directly from the data factorization results in Lemma 2,
link the lifting variables to the collected data and form
constraints \eqref{eq:opt_con3_MAS}–\eqref{eq:opt_con4_MAS}.

The optimization objective $\,-\log\det P_i\,$ is a convex surrogate
for maximizing the volume of
$\mathcal E_i(P_i,0)$, so the solution provides the largest
$\lambda$‑contractive ellipsoid contained in $\mathcal P_i$. This completes the proof.

\end{proof}

The semidefinite program \eqref{eq:opt_MAS} yields, for each agent, a contractive ellipsoid that is fully contained in the admissible polytope and invariant under the data-driven feedback gain \eqref{eq:gain_MAS_final}. In the next section (single agent), we evaluate this program at sampled waypoints to assemble a sequence of overlapping safe ellipsoids that respect dynamics and workspace constraints. The subsequent section extends this construction to multiple agents by adding a pairwise-safety test: a candidate waypoint is retained only if its (projected) ellipsoid does not intersect any active ellipsoid of the other agents. Together, these steps produce SAFE--MA--RRT plans that guide each agent from start to goal with certified safety at every step.

 \section{Single-Agent Safe RRT Design}
This section develops the single-agent planner that links the contractive ellipsoids from the previous section via sampling to produce safe, dynamics-aware paths. A candidate waypoint is accepted only if its ellipsoid overlaps that of its parent, yielding a chain of overlapping invariant ellipsoids from start to goal. The following section extends this construction to multiple agents by adding a pairwise-safety check: a candidate waypoint is retained only if its (projected) ellipsoid does not intersect any active ellipsoid from the other agents, thereby ensuring inter-agent safety. 

The traditional rapidly-exploring random tree (RRT) grows a tree in
state space by iteratively (i) sampling a random state
$q_{\text{rand}}$, (ii) selecting the nearest vertex $q_{\text{near}}$
in the current tree, (iii) steering a short step toward the sample to
obtain $q_{\text{new}}$, and (iv) adding the new vertex if the
straight-line segment $q_{\text{near}}\!\rightarrow\!q_{\text{new}}$
does not intersect static obstacles. While fast, this procedure ignores
system dynamics and offers no safety guarantees: straight-line steering
may be dynamically infeasible, and the path can leave the admissible set
or violate state constraints between waypoints.

In our variant, we first draw a continuous random sample
$q_{\mathrm{rand}} \in \mathcal{Y}$ with goal bias, and then snap it to
the grid, i.e., map it to a discrete cell: the nearest free cell center
(equivalently, the cell containing $q_{\mathrm{rand}}$ in a uniform
partition). This procedure yields a candidate waypoint
$q_{j} \in \mathcal{P}^{\mathrm{cand}}$. The tree’s vertices are
therefore candidate waypoints $q_{j}$ located at cell centers, and all
expansions operate over these cells (one–cell moves).

The proposed method augments classical RRT with two ingredients:

\begin{enumerate}[label=(\alph*)]
\item \textbf{Dynamics awareness via grid steering.}
      The workspace is discretized into square cells, and each tree extension advances by a single neighboring cell in the direction of the sampled target. This one–cell expansion aligns with single–step reachability and avoids large, dynamically infeasible jumps. Combined with the local feedback from the certificate, it yields motion segments that are consistent with the system dynamics.

\item \textbf{Local safety certificates (with overlap).}
      For each one–cell proposal from $c_{\text{near}}$ to $c_{\text{new}}$, we form the smallest axis–aligned rectangle covering the two cells and certify it via the data–driven SDP of Section~\ref{sec:inv_set}. A feasible solve returns a contractive ellipsoid fully contained in that rectangle together with a local feedback gain. The move is accepted only if this child ellipsoid also overlaps the parent ellipsoid,
      ensuring workspace safety and safe transition along waypoints.
\end{enumerate}

\subsection{Offline Path-Planning}
The single-agent safe RRT algorithm proceeds through the following steps.  
Each step is detailed below, from random sampling to the final acceptance of a new waypoint.

\subsubsection{Random sampling}
We draw a random point in free space with a goal bias:
\begin{align}
q_{\mathrm{rand}} &=
\begin{cases}
\textsc{Center}(c_g), & \text{with probability }\beta, \\[3pt]
\textsc{RandPoint}(\mathcal{G},\,\textit{blocked}), & \text{otherwise}.
\end{cases}
\label{eq:qrand}
\end{align}

Snap the point to its nearest free grid cell:
\begin{align}
c_{\mathrm{rand}}
&= \arg\min_{c \in \mathcal{G}_{\mathrm{free}}}
\bigl\| q_{\mathrm{rand}} - \textsc{Center}(c) \bigr\|.
\label{eq:crand}
\end{align}

\subsubsection{Nearest-vertex selection}
Pick the closest existing vertex in the $\ell_1$ metric:
\begin{align}
c_{\mathrm{near}}
&= \arg\min_{v \in V} \bigl\| v - c_{\mathrm{rand}} \bigr\|_{1}.
\label{eq:cnear}
\end{align}

\subsubsection{One-Cell Extension (No Cross-Grid Jumps)}
For any grid cell $c \in \mathcal{G}$, let $\mathcal{N}_4(c)$ denote its
4-neighborhood, i.e., the set of cells that share a horizontal or
vertical edge with $c$: $\mathcal{N}_4(c) \;=\; \{\,c' \in \mathcal{G} \;\mid\; \|c'-c\|_1=1\,\}$. From $c_{\mathrm{near}}$, \textsc{BestNeighbour} evaluates the four adjacent
cells (N/E/S/W), discards those outside the workspace or marked as blocked,
and returns
\begin{align}
c_{\mathrm{new}}
&=
\arg\min_{c\in \mathcal{N}_4(c_{\mathrm{near}})}
\lVert c - c_{\mathrm{rand}}\rVert_1 .
\label{eq:onecell}
\end{align}

The candidate $c_{\mathrm{new}}$ is rejected if
$\textit{blocked}(c_{\mathrm{new}})=\texttt{true}$ or if
$c_{\mathrm{new}}\in V$. Restricting to immediate neighbors ensures that
expansions advance by a single step on the grid, maintaining discrete-time
reachability and avoiding dynamically infeasible jumps.

\subsubsection{Local box and certificate (SDP)}
Form the axis-aligned rectangle $\mathcal P$ that covers the union
$c_{\mathrm{near}}\cup c_{\mathrm{new}}$ in $(x,y)$ and convert it to half-space form
$(F_{xy},g_{xy}) \in \mathbb{R}^{4\times 2}\times\mathbb{R}^4_{\ge 0}$.
Build full-state constraints by combining the position bounds with any additional
limits on non-position states (if any):
\begin{align}
F &=
\begin{bmatrix}
F_{xy}\,C \\[2pt]
F_{\text{extra}}
\end{bmatrix},
\quad
g =
\begin{bmatrix}
g_{xy} \\[2pt]
g_{\text{extra}}
\end{bmatrix}.
\label{eq:full_halfspace}
\end{align}

Solve the data-driven SDP with $(U_0,X_0,X_1)$, $(F,g)$, and contractivity~$\lambda$:
\begin{align}
(\mathrm{feas},\,P_{\mathrm{new}},\,K_{\mathrm{new}})
&= \texttt{SolveSDP}\bigl(U_0, X_0, X_1, F, g, \lambda\bigr).
\label{eq:solve_lmi}
\end{align}

If $\mathrm{feas}=\texttt{false}$, reject and resample.

\subsubsection{Overlap check and commit}
Let $c_{\mathrm{sh}}$ be the unique cell shared by the parent and child rectangles and define
\begin{equation}\label{eq:cmid}
c_{\mathrm{mid}} := \textsc{Center}\!\bigl(c_{\mathrm{sh}}\bigr).
\end{equation}

Project state-space ellipsoids to output space via their precisions:
\begin{equation}\label{eq:proj_precisions}
P_{\mathrm{proj,par}}^{-1} = C\,P_{\mathrm{par}}^{-1}\,C^{\!\top},
\qquad
P_{\mathrm{proj,new}}^{-1} = C\,P_{\mathrm{new}}^{-1}\,C^{\!\top}.
\end{equation}

Let the output centers at the two nodes be
\begin{equation}\label{eq:r_par_new}
r_{\mathrm{par}} := \textsc{Center}(c_{\mathrm{near}}), 
\qquad 
r_{\mathrm{new}} := \textsc{Center}(c_{\mathrm{new}}).
\end{equation}

Require the shared-cell center to lie in both projected ellipsoids, measured as deviations from each node’s center:
\begin{align}
(c_{\mathrm{mid}}-r_{\mathrm{par}})^{\!\top} P_{\mathrm{proj,par}}^{-1} (c_{\mathrm{mid}}-r_{\mathrm{par}}) &\le 1, 
\label{eq:proj_overlap_par}\\
(c_{\mathrm{mid}}-r_{\mathrm{new}})^{\!\top} P_{\mathrm{proj,new}}^{-1} (c_{\mathrm{mid}}-r_{\mathrm{new}}) &\le 1.
\label{eq:proj_overlap_new}
\end{align}
(At the root, skip~\eqref{eq:proj_overlap_par} if no parent certificate exists.)
If both are satisfied, commit the node and store
\begin{align}
V &\gets V \cup \{c_{\mathrm{new}}\}, \\
\textit{parent}(c_{\mathrm{new}}) & \gets c_{\mathrm{near}},
\label{eq:tree_update}\\
\bigl(P_{\mathrm{new}},\,K_{\mathrm{new}},\,r_{\mathrm{new}}\bigr)
&\;\text{attached to }(c_{\mathrm{near}}\!\to c_{\mathrm{new}}),
\\
r_{\mathrm{new}} & := \textsc{Center}(c_{\mathrm{new}}).
\label{eq:attach_cert}
\end{align}

The final safe path is obtained by backtracking from $c_g$ to $c_s$; the selected waypoints are the cell centers
$p_\ell := \textsc{Center}(c_\ell)$, with the stored per-edge certificates $(P,K)$.

Algorithm~\ref{alg:SA_RRT} summarizes the full sequence of operations for the single-agent safe RRT method described above.

\begin{algorithm}[h!]
\caption{\textsc{Single-Agent Safe–RRT Algorithm}}
\label{alg:SA_RRT}
\begin{algorithmic}[1]
\Require grid $\mathcal G$, obstacle mask \textit{blocked}; start $c_s$; goal $c_g$; goal–bias $\beta$; data $(U_0,X_0,X_1)$; contractivity $\lambda$; measurement matrix $C\in\mathbb{R}^{2\times n}$; optional constraints on non–position states; SDP solver \texttt{SolveLMI}
\State $V\gets\{c_s\}$,\; $\textit{parent}(c_s)\gets c_s$;\; attach initial $(P_s,K_s)$ if available
\While{$c_g\notin V$}
    \State \textbf{Sample point:}
    \Statex \hspace{\algorithmicindent}%
    $q_{\mathrm{rand}} \leftarrow
    \begin{cases}
      \textsc{Center}(c_g), & \text{w.p. }\beta,\\
      \textsc{RandPoint}(\mathcal{G},\textit{blocked}), & \text{o.w.}
    \end{cases}$
    \State \textbf{Snap to grid:}\;
           $c_{\mathrm{rand}} \gets \arg\min_{c\in\mathcal{G}_{\mathrm{free}}}
           \|q_{\mathrm{rand}} - \textsc{Center}(c)\|$
    \State \textbf{Nearest:}\;
           $c_{\mathrm{near}}\gets \arg\min_{v\in V}\|v-c_{\mathrm{rand}}\|_1$
    \State \textbf{Extend one cell:}\;
           $c_{\mathrm{new}}\gets\textsc{BestNeighbour}(c_{\mathrm{near}},c_{\mathrm{rand}})$
    \If{$\textit{blocked}(c_{\mathrm{new}})$ \textbf{or} $c_{\mathrm{new}}\in V$} \State \textbf{continue} \EndIf
    \State \textbf{Local box in position:}\; $\mathcal P \gets$ axis–aligned rectangle covering $c_{\mathrm{near}}\cup c_{\mathrm{new}}$
    \State $(F_{xy},g_{xy})\gets\textsc{BoxToHalfspace}(\mathcal P)$ \Comment{$F_{xy}\in\mathbb R^{4\times 2}$, $g_{xy}\in\mathbb R^4_{\ge 0}$}
    \State \textbf{Full–state halfspaces:} build $(F,g)$ by combining the position constraints with any additional bounds on other states, e.g.
           \[
              F \gets \begin{bmatrix} F_{xy}C \\ F_{\text{extra}} \end{bmatrix},\qquad
              g \gets \begin{bmatrix} g_{xy} \\ g_{\text{extra}} \end{bmatrix},
           \]
           where $(F_{\text{extra}},g_{\text{extra}})$ encode optional constraints (e.g., velocity limits).
    \State \textbf{Child certificate (SDP):}\;
           $(\text{feas},P_{\text{new}},K_{\text{new}})\gets\texttt{SolveSDP}(U_0,X_0,X_1,F,g,\lambda)$
    \If{\textbf{not} \text{feas}} \State \textbf{continue} \EndIf
    \State \textbf{Projected overlap test (in output space):}
           \Statex \hspace{\algorithmicindent}%
           $c_{\mathrm{mid}} \gets \textsc{Center}(\textsc{SharedCell}(c_{\mathrm{near}},c_{\mathrm{new}}))$;\;
           $P_{\text{par}}\gets P(c_{\mathrm{near}})$ (if available)\\
           \hspace{\algorithmicindent}%
           $P_{\text{proj,par}}^{-1}\gets C\,P_{\text{par}}^{-1}C^\top$;\;
           $P_{\text{proj,new}}^{-1}\gets C\,P_{\text{new}}^{-1}C^\top$
    \If{$P_{\text{par}}$ exists \textbf{and} $\bigl(c_{\mathrm{mid}}^\top P_{\text{proj,par}}^{-1} c_{\mathrm{mid}}>1$ \textbf{or} $c_{\mathrm{mid}}^\top P_{\text{proj,new}}^{-1} c_{\mathrm{mid}}>1\bigr)$}
        \State \textbf{continue} \Comment{shared cell center not contained in both projected ellipsoids}
    \EndIf
    \State \textbf{Accept:}\; $V\gets V\cup\{c_{\mathrm{new}}\}$;\;
           $\textit{parent}(c_{\mathrm{new}})\gets c_{\mathrm{near}}$;\;
           attach $\bigl(P_{\text{new}},K_{\text{new}},\,r=\textsc{Center}(c_{\mathrm{new}})\bigr)$ to $c_{\mathrm{new}}$
\EndWhile
\State \Return safe path $\{p_\ell\}_{\ell=0}^{N_w}$ with 
$p_\ell = \textsc{Center}(c_\ell)$ obtained by backtracking 
from $c_g$ to $c_s$, together with per–edge $(P,K)$.
\end{algorithmic}
\end{algorithm}

\subsection{Online Execution}
Given a certified cell sequence 
\(\pi=\langle c_0,\dots,c_{N_w}\rangle\) returned by 
Algorithm~\ref{alg:SA_RRT}, we execute it segment by segment with the 
local data–driven affine controllers. The corresponding selected 
waypoints are defined as 
\begin{equation}
    p_\ell := \mathrm{Center}(c_\ell), \qquad \ell=0,\dots,N_w.
\end{equation}

For each edge \((c_{\ell-1}\!\to c_\ell)\) we set the reference to 
\(p_\ell\) and compute the steady state 
\((\bar x_\ell,\bar u_\ell)\) via the data–based map 
in~\eqref{eq.equilibrium_calculation}, which enforces 
\(C\bar x_\ell = p_\ell\). 
While traversing that segment, we apply the control law~\eqref{eq.controller_final} 
with the certificate \((P_\ell,K_\ell)\) attached to the edge, and monitor safety in 
the measured output \(y=Cx\) by testing membership in the projected ellipsoid 
\(\mathcal E_y\!\bigl(P_{\mathrm{proj},\ell},\,p_\ell\bigr)\), where 
\(P_{\mathrm{proj},\ell}^{-1}=C\,P_\ell^{-1}C^\top\). 
We switch to the next controller as soon as \(y\) enters the overlapping next 
projected ellipsoid \(\mathcal E_y\!\bigl(P_{\mathrm{proj},\ell+1},\,p_{\ell+1}\bigr)\). 
This model–free rollout preserves planning–time safety guarantees because 
handoffs occur only within certified overlaps. 
Algorithm~\ref{alg:exec_single} summarizes the execution procedure.

\begin{algorithm}[h!]
\caption{\textsc{Execute Certified Path} (Single-Agent)}
\label{alg:exec_single}
\begin{algorithmic}[1]
\Require Backtracked cell path $\pi=\langle c_0,\dots,c_{N_w}\rangle$; per-edge certificates $\{(P_\ell,K_\ell)\}_{\ell=1}^{N_w}$ (each attached to edge $(c_{\ell-1}\!\to c_\ell)$); output map $C$; steady–state map $\hat T$ from \eqref{eq.equilibrium_calculation}; initial state $x_0$; tolerance $r_f>0$
\Ensure Safe execution from $p_0=\mathrm{Center}(c_0)$ to $p_{N_w}=\mathrm{Center}(c_{N_w})$ in the measured output $y=Cx$
\State \textbf{Define waypoints:} for $\ell=0,\dots,N_w$, set $p_\ell \gets \mathrm{Center}(c_\ell)$
\State \textbf{Precompute steady states:} for $\ell=1,\dots,N_w$,
\Statex \hspace{\algorithmicindent}
$\displaystyle
\begin{bmatrix}\bar x_\ell \\ \bar u_\ell\end{bmatrix}
\gets
\hat T^{-1}
\begin{bmatrix} 0 \\ p_\ell \end{bmatrix}$ 
\Comment{equilibrium enforcing $C\bar x_\ell=p_\ell$}
\State \textbf{Project certificates to output space:} for $\ell=1,\dots,N_w$,
\Statex \hspace{\algorithmicindent}
$P_{\mathrm{proj},\ell}^{-1} \gets C\,P_\ell^{-1}\,C^\top$
\Statex \hspace{\algorithmicindent}\Comment{$\mathcal E_y(P_{\mathrm{proj},\ell},p_\ell)=\{y:(y-p_\ell)^\top P_{\mathrm{proj},\ell}^{-1}(y-p_\ell)\le 1\}$}
\State $\ell \gets 1$, \; $k \gets 0$
\While{true}
    \State \textbf{Control law:} \quad $u_k \gets K_\ell\bigl(x_k-\bar x_\ell\bigr) + \bar u_\ell$
    \State \textbf{State update:} \quad apply $u_k$ and obtain $x_{k+1}$ \;(e.g., $x_{k+1}=Ax_k+Bu_k$ or measured); set $y_{k+1}\gets Cx_{k+1}$
    \State \textbf{Safety check (output ellipsoid):}
    \Statex \hspace{\algorithmicindent}
    \textbf{if } $(y_{k+1}-p_\ell)^\top P_{\mathrm{proj},\ell}^{-1}(y_{k+1}-p_\ell) > 1$
    \textbf{ then abort (safety violated)}
    \If{$\ell<N_w$ \textbf{and} $(y_{k+1}-p_{\ell+1})^\top P_{\mathrm{proj},\ell+1}^{-1}(y_{k+1}-p_{\ell+1}) \le 1$}
        \State $\ell \gets \ell+1$ \Comment{handoff when entering next certified output ellipsoid}
    \ElsIf{$\ell=N_w$ \textbf{and} $\|y_{k+1}-p_{N_w}\|_2 \le r_f$}
        \State \Return executed trajectory $\{x_t\}_{t=0}^{k+1}$ (and outputs $\{y_t\}$)
    \EndIf
    \State $k \gets k+1$
\EndWhile
\end{algorithmic}
\end{algorithm}

\section{Multi-Agent Safe RRT Design}
The single-agent safe RRT grows one dynamics-aware tree whose
edges are certified safe by the SDP of Section~\ref{sec:inv_set}. To
extend this idea to a team of $N_a$ agents that advance on
the same discrete time grid, all trees are expanded
synchronously and coordinated through a global space–time
reservation table. The array $\textit{res}[k]$ (one Boolean grid per
layer $k\in\mathbb{N}_0$) records which cells are already claimed at
time~$k$. Each vertex $v$ in agent~$i$’s tree carries an integer
$\textit{depth}_i(v)$ equal to the number of one-cell moves from the
start, so $\textit{depth}_i(v)$ coincides with the time layer at which
that cell will be reached if every accepted edge is executed in one
sampling period. Inter-agent safety is enforced by rejecting any
proposal that conflicts in space–time with existing reservations, as
well as proposals that would cause swaps (head-on exchanges). While our coordination borrows synchronized expansion and reservations from cooperative pathfinding \cite{silver2005cooperative}, those geometric methods do not provide per-edge, data-driven invariant certificates; here each accepted move is certified by an SDP ellipsoid and an overlap-based switching rule. 

\subsection{Offline Path Planning}
The steps for offline path planning in the multi-agent setting are as follows:

\subsubsection{Random Sampling}
Each agent draws a continuous exploration point, then maps it to its containing grid cell. With goal bias~$\beta$,
\begin{align}
q^i_{\mathrm{rand}} &=
\begin{cases}
\mathrm{Center}(g^{i}), & \text{with probability }\beta,\\
\textsc{RandPoint}\!\big(\mathcal{X}_{\mathrm{free}}\big), & \text{otherwise},
\end{cases}
\label{eq:mas_qrand}\\[3pt]
c^i_{\mathrm{rand}}
&=
\arg\min_{c\in\mathcal G_{\mathrm{free}}}
\left\|q^i_{\mathrm{rand}}-\mathrm{Center}(c)\right\|_2 .
\label{eq:mas_crand}
\end{align}

Refreshing the target at each layer preserves the rapid-exploration property of RRT.

\subsubsection{Nearest-Vertex Search}
Given $c^i_{\mathrm{rand}}$, agent $i$ selects the existing vertex closest in the Manhattan metric:
\begin{align}
c^i_{\mathrm{near}}
&=
\arg\min_{v\in V_i}\,\lVert v-c^i_{\mathrm{rand}}\rVert_1 .
\label{eq:mas_nearest}
\end{align}

The $\ell_1$ metric matches the rectilinear grid and 4-neighbor connectivity; its cost is negligible relative to the SDP solve.

\subsubsection{One-Cell Extension (No Cross-Grid Jumps)}
From $c^i_{\mathrm{near}}$, \textsc{BestNeighbour} evaluates the four adjacent cells (N/E/S/W), discards those outside the workspace or inside debris, and returns
\begin{align}
c^i_{\mathrm{new}}
&=
\arg\min_{c\in \mathcal{N}_4(c^i_{\mathrm{near}})}
\lVert c - c^i_{\mathrm{rand}}\rVert_1 ,
\label{eq:mas_extend}\\
k_{\mathrm{new}}
&=
\textit{depth}_i\!\bigl(c^i_{\mathrm{near}}\bigr)+1
\;=\; k{+}1 .
\label{eq:mas_tnew}
\end{align}

Restricting to immediate neighbors aligns the discrete step with single-step reachability and prevents large, dynamically infeasible jumps.

\subsubsection{Child Certificate (SDP) and Overlap in Output Space}
Form the axis-aligned rectangle $\mathcal P^{\,i}$ covering $c^i_{\mathrm{near}}\cup c^i_{\mathrm{new}}$ in $(x,y)$, convert to half-space form, then lift to the full state if additional state constraints apply
\begin{align}
(F^{xy},g^{xy}) \;&=\; \textsc{BoxToHalfspace}\!\bigl(\mathcal P^{\,i}\bigr),
\label{eq:mas_box}\\
(F^{\,i},g^{\,i}) \;&=\; \textsc{LiftToState}\!\bigl(F^{xy},g^{xy}\bigr).
\label{eq:mas_lift}
\end{align}

Solve the data-driven SDP with contractivity $\lambda$ to obtain a local controller and contractive ellipsoid
\begin{align}
(\mathrm{feas},P^i_\ell,K^i_\ell)
\;&=\;
\texttt{SolveSDP}\!\bigl(U_{0,i},X_{0,i},X_{1,i},F^{\,i},g^{\,i},\lambda\bigr).
\label{eq:mas_sdp}
\end{align}

If \eqref{eq:mas_sdp} is infeasible, discard the candidate.

To guarantee safe switching, check overlap in the measured output $y_i(k)=C\,x_i(k)$ by testing the center of the shared cell in the projected ellipsoids
\begin{align}
c^i_{\mathrm{mid}} \;=\; \mathrm{Center}\!\Big(&\textsc{SharedCell}\big(c^i_{\mathrm{near}},c^i_{\mathrm{new}}\big)\Big),
\label{eq:mas_cmid}\\
(P^i_{\mathrm{par,proj}})^{-1} \;&=\; C\,(P^i_{\mathrm{par}})^{-1}C^\top, 
\\
(P^i_{\mathrm{new,proj}})^{-1} \;&=\; C\,(P^i_\ell)^{-1}C^\top,
\label{eq:mas_proj}\\
(c^i_{\mathrm{mid}})^\top &(P^i_{\mathrm{par,proj}})^{-1} c^i_{\mathrm{mid}} \;\le\; 1,
\\
(c^i_{\mathrm{mid}})^\top &(P^i_{\mathrm{new,proj}})^{-1} c^i_{\mathrm{mid}} \;\le\; 1.
\label{eq:mas_overlap_proj}
\end{align}

The candidate is accepted only if both inequalities in \eqref{eq:mas_overlap_proj} hold.\footnote{At a root node with no parent certificate, the overlap test is skipped.}

\subsubsection{Space--Time Conflict Resolution}
Collect all proposals that passed \eqref{eq:mas_sdp}--\eqref{eq:mas_overlap_proj} into $\mathcal C_k$ and prune by:
\begin{itemize}
\item \emph{Reservation conflict:} reject any proposal with $\textit{res}[k{+}1](c^i_{\mathrm{new}})=\texttt{true}$;
\item \emph{Same-cell collision:} if multiple agents propose the same cell at layer $k{+}1$, keep one (e.g., smaller $\ell_1$ heuristic-to-go) and break ties randomly;
\item \emph{Swap conflict:} if a head-on swap occurs,
$(c^i_{\mathrm{near}},c^i_{\mathrm{new}})=(c^j_{\mathrm{new}},c^j_{\mathrm{near}})$, reject one via the same rule.
\end{itemize}

\subsubsection{Synchronous Commit and Reservation Update}
Commit all surviving proposals simultaneously. For each accepted agent $i$,
\begin{align}
&V_i \gets V_i \cup \{c^i_{\mathrm{new}}\}, \qquad
\textit{parent}(c^i_{\mathrm{new}}) \gets c^i_{\mathrm{near}},
\label{eq:mas_commit1}\\
&\textit{depth}_i(c^i_{\mathrm{new}}) \gets k{+}1, \qquad
\textit{res}[k{+}1](c^i_{\mathrm{new}}) \gets \texttt{true},
\label{eq:mas_commit2}\\
&\text{attach } (P^i_\ell,K^i_\ell) \text{ to edge } (c^i_{\mathrm{near}}\!\to c^i_{\mathrm{new}}).
\label{eq:mas_commit3}
\end{align}

Iterating these stages grows all trees in synchronized layers until every agent reaches its goal. Because each accepted edge carries its own data-driven feedback gain and contractive ellipsoid, the resulting space--time trajectories are dynamics-aware and respect workspace constraints without explicit inter-agent distance checks. Algorithm~\ref{alg:MA_RRT} summarizes the simultaneous multi-agent planning procedure.

\begin{algorithm}[h!]
\caption{\textsc{SAFE--MA--RRT}}
\label{alg:MA_RRT}
\begin{algorithmic}[1]
\Require Grid $\mathcal G$, mask \emph{blocked}; starts $\{c^i_0\}$, goals $\{c^i_{N_w^{(i)}}\}$; goal-bias $\beta$; data $\{U_{0,i},X_{0,i},X_{1,i}\}$; contractivity $\lambda$; measurement matrix $C\in\mathbb{R}^{2\times n}$; SDP solver \texttt{SolveLMI}
\State \textbf{Init:} For all $i$, set $V_i\gets\{c^i_0\}$, $\textit{parent}(c^i_0)\gets c^i_0$, $\textit{depth}_i(c^i_0)\gets0$; attach initial certificate to $c^i_0$ if available
\State $\textit{res}[0]\gets\texttt{false}$; set $\textit{res}[0](c^i_0)\gets\texttt{true}$ for all $i$
\For{$k=0,1,2,\dots$ \textbf{until} all $c^i_{N_w^{(i)}}$ reached}
   \State $\mathcal C_k \gets \emptyset$
   \ForAll{$i$ with $c^i_{N_w^{(i)}}\notin V_i$}
      \State \textbf{Sample--snap--nearest:} draw $q^i_{\mathrm{rand}}$ (goal bias $\beta$), snap to $c^i_{\mathrm{rand}}\in\mathcal G_{\mathrm{free}}$, and select $c^i_{\mathrm{near}}\in V_i$ minimizing $\|v-c^i_{\mathrm{rand}}\|_1$
      \State \textbf{Extend one cell:}\;
             $c^i_{\mathrm{new}}\gets \textsc{BestNeighbour}(c^i_{\mathrm{near}},c^i_{\mathrm{rand}})$
      \If{\emph{blocked}($c^i_{\mathrm{new}}$) $\lor$ $c^i_{\mathrm{new}}\in V_i$} \State \textbf{continue} \EndIf
      \State \textbf{Local box (in $(x,y)$) and full-state half-spaces:}
      \Statex \hspace{\algorithmicindent}%
             $\mathcal P^i\gets$ rectangle covering $c^i_{\mathrm{near}}\cup c^i_{\mathrm{new}}$;\;
             $(F^i_{xy},g^i_{xy})\gets\textsc{BoxToHalfspace}(\mathcal P^i)$;\;
             $(F^i,g^i)\gets\textsc{LiftToState}(F^i_{xy},g^i_{xy})$
      \State \textbf{Child certificate (SDP):}\;
             $(\text{feas},P^i_\ell,K^i_\ell)\gets \texttt{SolveLMI}(U_{0,i},X_{0,i},X_{1,i},F^i,g^i,\lambda)$
      \If{\textbf{not} \text{feas}} \State \textbf{continue} \EndIf
      \State \textbf{Overlap test (projected to $(x,y)$):}
      \Statex \hspace{\algorithmicindent}%
             $c^i_{\mathrm{mid}} \gets \mathrm{Center}\!\bigl(\textsc{SharedCell}(c^i_{\mathrm{near}},c^i_{\mathrm{new}})\bigr)$,
             \quad $P^i_{\mathrm{par}}\gets\text{cert}\!\bigl(c^i_{\mathrm{near}}\bigr)$
      \If{$P^i_{\mathrm{par}}$ exists}
         \State $(P^i_{\mathrm{proj,par}})^{-1} \gets C\,(P^i_{\mathrm{par}})^{-1}C^\top$, \quad
                $(P^i_{\mathrm{proj,new}})^{-1} \gets C\,(P^i_\ell)^{-1}C^\top$
         \If{$(c^i_{\mathrm{mid}})^\top (P^i_{\mathrm{proj,par}})^{-1} c^i_{\mathrm{mid}} > 1$
             \textbf{or} $(c^i_{\mathrm{mid}})^\top (P^i_{\mathrm{proj,new}})^{-1} c^i_{\mathrm{mid}} > 1$}
             \State \textbf{continue} \Comment{shared-cell center not in both projected ellipsoids}
         \EndIf
      \EndIf
      \State Add proposal $\bigl(i, c^i_{\mathrm{near}}, c^i_{\mathrm{new}}, P^i_\ell,K^i_\ell\bigr)$ to $\mathcal C_k$
   \EndFor
   \State \textbf{Conflict resolution (space--time):} discard proposals violating $\textit{res}[k{+}1]$; for duplicates or swaps, keep one by \textsc{TieBreak}
   \State \textbf{Commit (synchronous):} For each accepted $(i,c^i_{\mathrm{near}},c^i_{\mathrm{new}},P^i_\ell,K^i_\ell)$:
   \Statex \hspace{\algorithmicindent}%
      $V_i\gets V_i\cup\{c^i_{\mathrm{new}}\}$;\;
      $\textit{parent}(c^i_{\mathrm{new}})\gets c^i_{\mathrm{near}}$;\;
      $\textit{depth}_i(c^i_{\mathrm{new}})\gets k{+}1$;\;
      $\textit{res}[k{+}1](c^i_{\mathrm{new}})\gets\texttt{true}$;\;
      attach $(P^i_\ell,K^i_\ell)$ to edge $(c^i_{\mathrm{near}}\!\to c^i_{\mathrm{new}})$
\EndFor
\State \Return $\{\pi^i\}$ by backtracking each agent’s parent; each edge carries $(P^i_\ell,K^i_\ell)$
\end{algorithmic}
\end{algorithm}

\subsection{Online Execution}
Given the certified paths $\{\pi^{i}\}$ returned by the planner and the per–edge certificates $\{(P^{i}_\ell,K^{i}_\ell)\}$, execution proceeds segment–by–segment and in lockstep across agents. For each edge $(c^{i}_{\ell-1}\!\to c^{i}_{\ell})$ we define the selected waypoint $p^i_\ell := \mathrm{Center}(c^i_\ell)$, compute the steady state $(\bar x^{\,i}_\ell,\bar u^{\,i}_\ell)$ from the data–based map in~\eqref{eq.equilibrium_calculation} such that $C\,\bar x^{\,i}_\ell = p^i_\ell$, and apply the affine law~\eqref{eq.controller_final} using $(P^{i}_\ell,K^{i}_\ell)$. Safety is monitored in the measured output $y_i(k)=C\,x_i(k)$ by testing membership in the projected ellipsoid $\mathcal E_y(P^{i}_{\mathrm{proj},\ell},p^i_\ell)=\{\,y : (y-p^i_\ell)^\top (P^{i}_{\mathrm{proj},\ell})^{-1} (y-p^i_\ell)\le 1\,\}$, where $(P^{i}_{\mathrm{proj},\ell})^{-1}=C\,(P^{i}_\ell)^{-1}C^\top$. Each agent switches to the next controller as soon as its output enters the overlapping projected ellipsoid around $p^i_{\ell+1}$; these synchronized handoffs preserve the planning–time safety guarantees for all agents. Algorithm~\ref{alg:exec_multi} summarizes the procedure.

\begin{algorithm}[h!]
\caption{\textsc{Execute Certified Paths} (Multi-Agent)}
\label{alg:exec_multi}
\begin{algorithmic}[1]
\Require Paths $\{\pi^{i}\}$ with $\pi^{i}=\langle c^{i}_0,\dots,c^{i}_{N^{(i)}_w}\rangle$; per-edge certificates $\{(P^{i}_\ell,K^{i}_\ell)\}_{\ell=1}^{N^{(i)}_w}$ (each attached to edge $(c^{i}_{\ell-1}\!\to c^{i}_\ell)$); steady--state map $\hat T$ from~\eqref{eq.equilibrium_calculation}; initial states $\{x_i(0)\}$; tolerance $r_f>0$; output map $C$
\Ensure Safe, synchronized execution for all agents under output-space constraints ($y_i(k)=C\,x_i(k)$)
\State \textbf{Define waypoints:} for each agent $i$ and $\ell=0,\dots,N^{(i)}_w$, set $p^i_\ell \gets \mathrm{Center}(c^i_\ell)$
\State \textbf{Precompute steady states:} for each $i$, $\ell=1,\dots,N^{(i)}_w$,
\[
\begin{bmatrix}\bar x^{\,i}_\ell \\ \bar u^{\,i}_\ell\end{bmatrix}
\gets
\hat T^{-1}
\begin{bmatrix} 0 \\ p^i_\ell \end{bmatrix},
\qquad C\,\bar x^{\,i}_\ell = p^i_\ell
\]
\State \textbf{Precompute projected precisions:} for each $i$, $\ell=1,\dots,N^{(i)}_w$,
\[
\bigl(P^{i}_{\mathrm{proj},\ell}\bigr)^{-1} \gets C\,\bigl(P^{i}_\ell\bigr)^{-1}C^{\!\top}
\]
\State For each agent $i$, initialize its current segment index $\ell \gets 1$; set global time $k\gets 0$
\While{some agent not finished}
  \ForAll{agents $i$ not finished}
    \State \textbf{Control law \eqref{eq.controller_final}:}
    \[
    u_i(k) \;=\; K^{i}_\ell\bigl(x_i(k)-\bar x^{\,i}_\ell\bigr) + \bar u^{\,i}_\ell
    \]
    \State \textbf{State/output update:} apply $u_i(k)$ and obtain $x_i(k{+}1)$ (e.g., $x_i(k{+}1)=A\,x_i(k)+B\,u_i(k)$ or measured); set $y_i(k{+}1)\gets C\,x_i(k{+}1)$
    \State \textbf{Safety in active projected ellipsoid:}
    \[
    \begin{aligned}
        \textbf{if }\;&
        \bigl(y_i(k{+}1)-p^i_\ell\bigr)^{\!\top}
        \bigl(P^{i}_{\mathrm{proj},\ell}\bigr)^{-1}
        \bigl(y_i(k{+}1)-p^i_\ell\bigr) \;>\; 1 \\
        &\textbf{then abort}
    \end{aligned}
    \]
    \If{$\ell<N^{(i)}_w$ \textbf{and}
        $(y_i(k{+}1)-p^i_{\ell+1})^\top
         \bigl(P^{i}_{\mathrm{proj},\,\ell+1}\bigr)^{-1}
         (y_i(k{+}1)-p^i_{\ell+1}) \le 1$}
        \State $\ell \gets \ell+1$ \Comment{handoff when entering next projected ellipsoid}
    \ElsIf{$\ell=N^{(i)}_w$ \textbf{and} $\|y_i(k{+}1)-p^i_{N^{(i)}_w}\|_2 \le r_f$}
        \State mark agent $i$ as finished
    \EndIf
  \EndFor
  \State $k\gets k+1$
\EndWhile
\State \Return executed trajectories $\bigl\{\{x_i(t)\}_{t=0}^{k_i}\bigr\}$ (and outputs $\{y_i(t)\}$ if desired)
\end{algorithmic}
\end{algorithm}

\begin{remark}
While our simultaneous expansion and reservation-table mechanism is inspired by cooperative pathfinding \cite{silver2005cooperative}, those techniques do not directly accommodate unknown dynamics or certify safety under feedback. In our approach, the geometric feasibility checks are replaced by data-driven invariant-set certificates (per edge) and an ellipsoid-overlap switching rule, thereby integrating dynamics awareness and safety guarantees into the MAS planner.
\end{remark}

\section{Simulation Studies}\label{sec:simulation}
This section evaluates the proposed SAFE--MA-RRT on the spacecraft rendezvous model, restricted to in–plane ($x$–$y$) motion. We first describe the dynamics and environment, then present (i) a single–spacecraft (single-agent) result and (ii) a two–spacecraft (multi-agent) result.

\subsection{System Dynamics and Environment}
We adopt the standard linearized relative-motion (Clohessy–Wiltshire) model in the local orbital frame and restrict attention to the in–plane coordinates. See \cite{guglieri2014design} for application context. Let $z_1,z_2$ denote the in–plane relative positions and $\dot z_1,\dot z_2$ their velocities; the inputs $v_1,v_2$ are in–plane thrust accelerations. The continuous-time dynamics are \cite{wie1998space}
\begin{equation}
\begin{aligned}
\ddot z_1 &= 3r^2 z_1 + 2r\,\dot z_2 + v_1,\\
\ddot z_2 &= -2r\,\dot z_1 + v_2,
\end{aligned}
\label{eq:CW_ct}
\end{equation}
with mean motion $r=1.1\times 10^{-1}\,\mathrm{s}^{-1}$. Stacking $x=[z_1\;z_2\;\dot z_1\;\dot z_2]^\top$ and $u=[v_1\;v_2]^\top$ gives
\begin{equation}
\dot x \;=\; A_c x + B_c u ,
\end{equation}
\[
A_c=\begin{bmatrix}
0 & 0 & 1 & 0\\
0 & 0 & 0 & 1\\
3r^2 & 0 & 0 & 2r\\
0 & 0 & -2r & 0
\end{bmatrix},
\qquad
B_c=\begin{bmatrix}
0 & 0\\
0 & 0\\
1 & 0\\
0 & 1
\end{bmatrix}.
\]

With a zero-order hold (ZOH) and sampling time $T_s=30\,\mathrm{s}$, the discrete-time system is
\begin{equation}
x(k{+}1) \;=\; A\,x(k) + B\,u(k),
\qquad
(A,B)=\mathrm{zoh}\bigl(A_c,B_c,T_s\bigr).
\end{equation}

The measured/output coordinates are the in–plane positions
\[
C=\begin{bmatrix}
1 & 0 & 0 & 0\\[2pt]
0 & 1 & 0 & 0
\end{bmatrix}.
\]

The planner operates in a $100\times100$\,m square workspace discretized into $10$\,m cells over $[-50,50]\times[-50,50]$\,m. Debris is modeled as $16$\,m, axis-aligned squares; cells overlapped by a debris object are marked blocked. For every candidate edge, a local state-feedback gain and a contractive ellipsoid are computed by solving the SDP of Section~\ref{sec:inv_set} with contraction factor $\lambda=0.94$ and a log-det objective. A move is accepted only if the new ellipsoid overlaps the parent ellipsoid at the shared mid-edge point. The goal-bias is $\beta=0.20$.

\subsection{Single-Agent Scenario}
The map contains a single $16$\,m square debris centered at the origin. The spacecraft starts at $(-45,-45)$\,m and must reach $(45,45)$\,m. At each iteration, the algorithm (i) samples a target cell with goal bias $\beta$, (ii) selects the nearest tree vertex in the $\ell_1$ metric, (iii) proposes a one-cell extension, and (iv) solves the SDP on the two-cell rectangle enclosing the parent–child pair. If the SDP is feasible and the mid-edge overlap holds, the new vertex is accepted and the associated $(P,K)$ certificate is stored.

\begin{figure}[h!]
    \centering
    \subcaptionbox{Planned Safe–RRT path with invariant ellipsoids.\label{fig:SA_path}}
        {\includegraphics[width=0.9\linewidth]{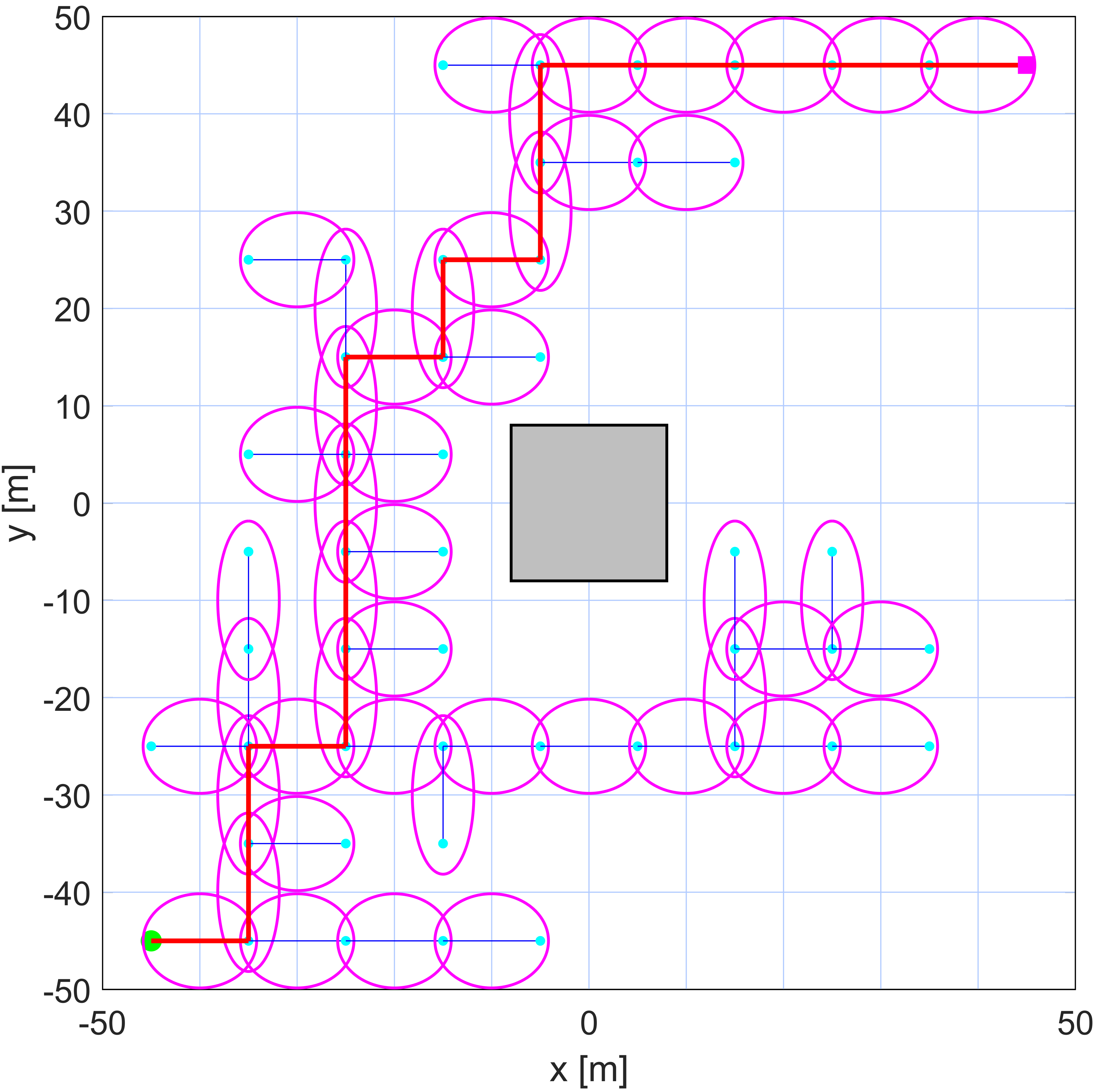}}\par\medskip
    \subcaptionbox{Executed closed-loop trajectory under per-edge controllers.\label{fig:SA_exec}}
        {\includegraphics[width=0.9\linewidth]{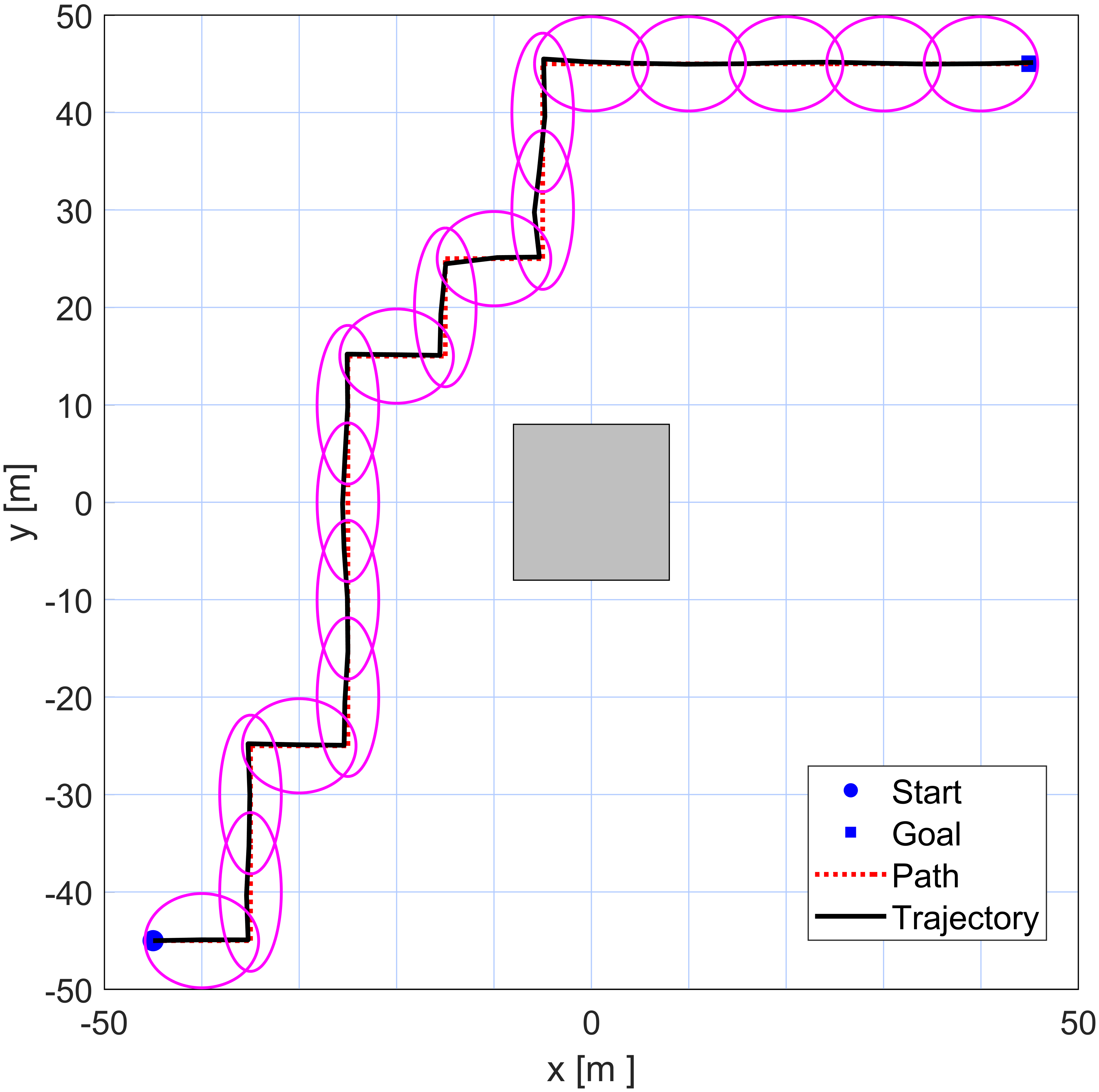}}
    \caption{Single-agent Safe–RRT in the $x$–$y$ plane. 
             (a) The magenta ellipses are contractive certificates covering each edge of the planned path. Overlap of consecutive ellipsoids guarantees safe controller switching and constraint satisfaction while avoiding debris. 
             (b) The red dashed line shows the planned waypoint path, while the black solid line shows the executed closed-loop trajectory under per-edge affine state-feedback gains, verifying feasibility under system dynamics.}
    \label{fig:SA_results}
\end{figure}

Figure~\ref{fig:SA_results} presents a representative single-agent run. In subplot (a), each edge of the Safe–RRT path is certified by a local contractive ellipsoid that renders its two-cell rectangle forward-invariant. The overlap between consecutive ellipsoids ensures safe switching and preserves admissibility. Subplot (b) shows the executed trajectory under the synthesized per-edge controllers, confirming that the system remains within constraints and successfully avoids the debris region.

\subsection{Two-Agent Scenario}
The grid resolution remains $10$\,m, but the map now contains seven
$16$\,m square debris fields centered at
$(-30,40)$, $(-40,-30)$, $(30,30)$, $(40,-20)$, $(-30,10)$,
$(10,-30)$, and the origin. Agent~A travels from $(-45,-45)$\,m to
$(45,45)$\,m, and Agent~B from $(-45,45)$\,m to $(45,-45)$\,m. All
trees expand synchronously on a common time index: in each
global iteration both agents propose one-cell moves for the next
layer, certify them via the SDP, verify overlap with their respective
parent ellipsoids, and then pass space–time checks—no double booking
of a cell at the same layer and no head-on swaps. Proposals that
survive all tests are committed atomically, and the reservation table
for that layer is updated.

\begin{figure}[h!]
    \centering
    \subcaptionbox{Exploration trees at termination (blue = A, red = B).\label{fig:MA_tree}}
        {\includegraphics[width=0.9\linewidth]{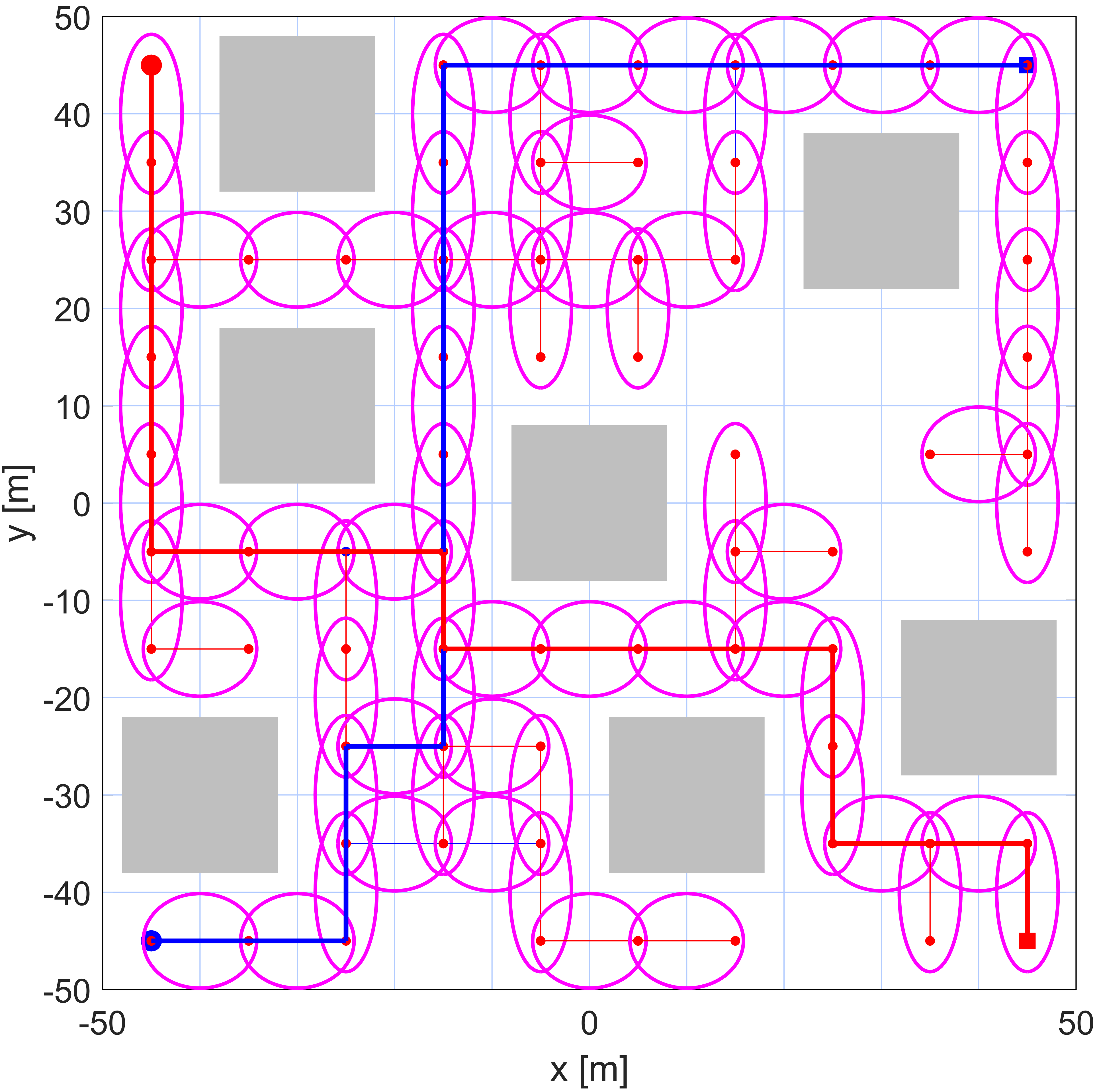}}\par\medskip
    \subcaptionbox{Planned paths with their associated contractive ellipsoids. The blue and red certificates do not overlap in time–space layers, thereby guaranteeing safety and pairwise separation throughout execution.\label{fig:MA_ell}}
        {\includegraphics[width=0.9\linewidth]{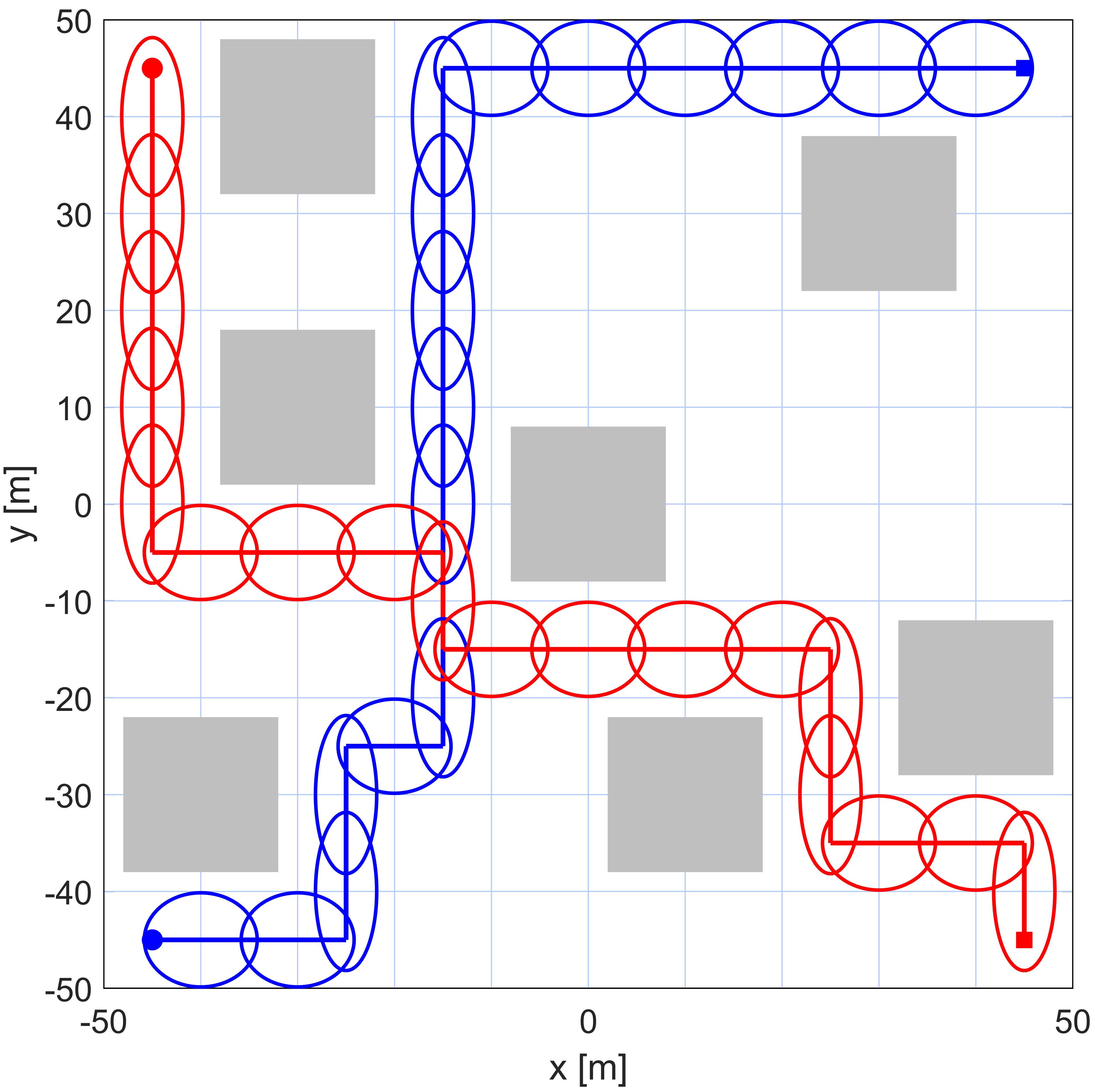}}
    \caption{Two-agent SAFE--MA--RRT: exploration and safety certificates.}
    \label{fig:MA_results_part1}
\end{figure}

\begin{figure}[h!]
    \centering
    \includegraphics[width=0.9\linewidth]{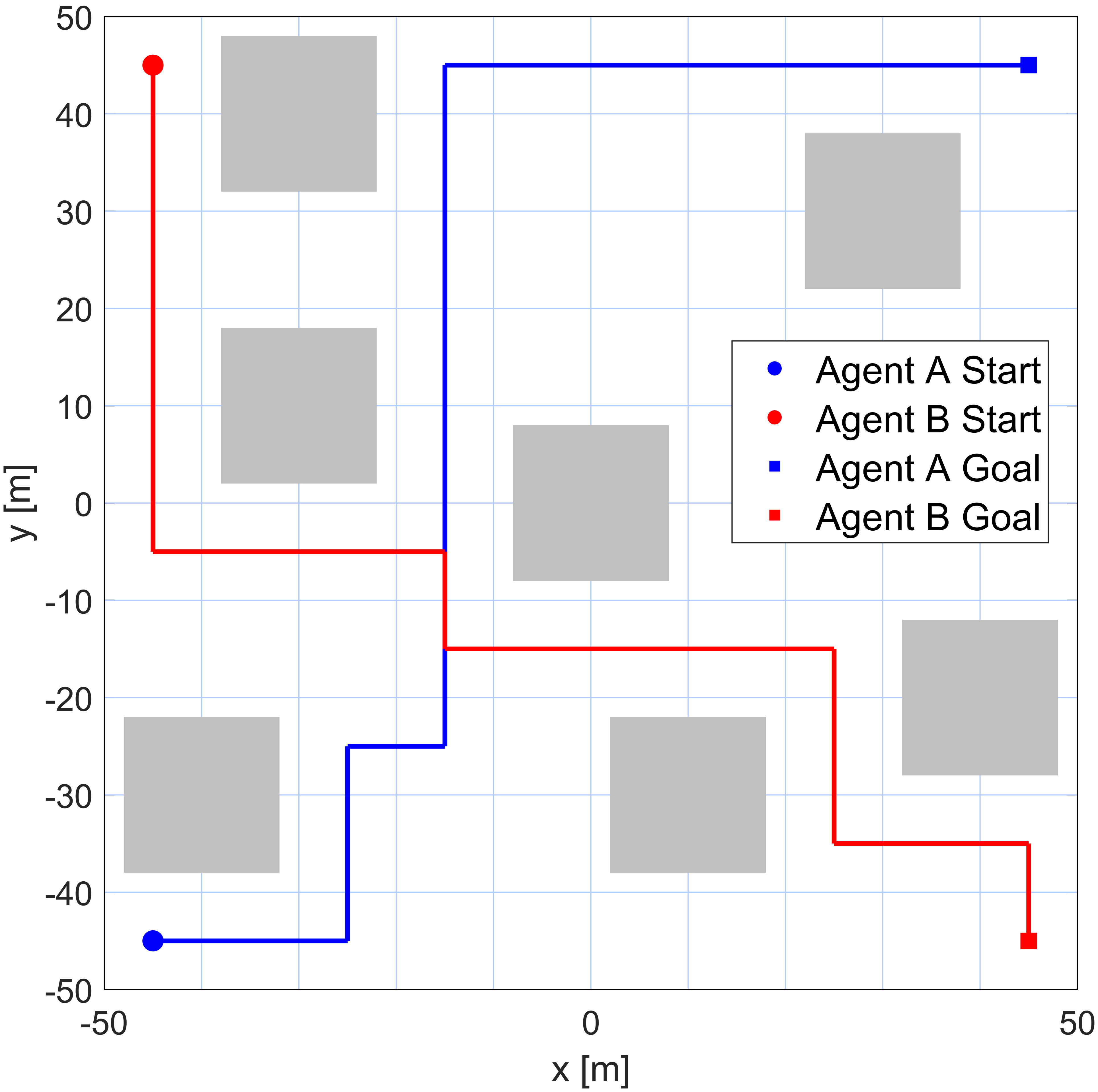}
    \caption{Two-agent SAFE--MA--RRT: planned waypoint paths. Each edge between adjacent cells is certified by a contractive ellipsoid; paths are shown without state traces for clarity.}
    \label{fig:MA_paths_plan_sub}
\end{figure}

\begin{figure}[h!]
    \centering
    \subcaptionbox{SAFE--MA--RRT executed trajectories under the data-driven affine controllers associated with each certified edge.\label{fig:MA_path_exec_safe}}
        {\includegraphics[width=0.9\linewidth]{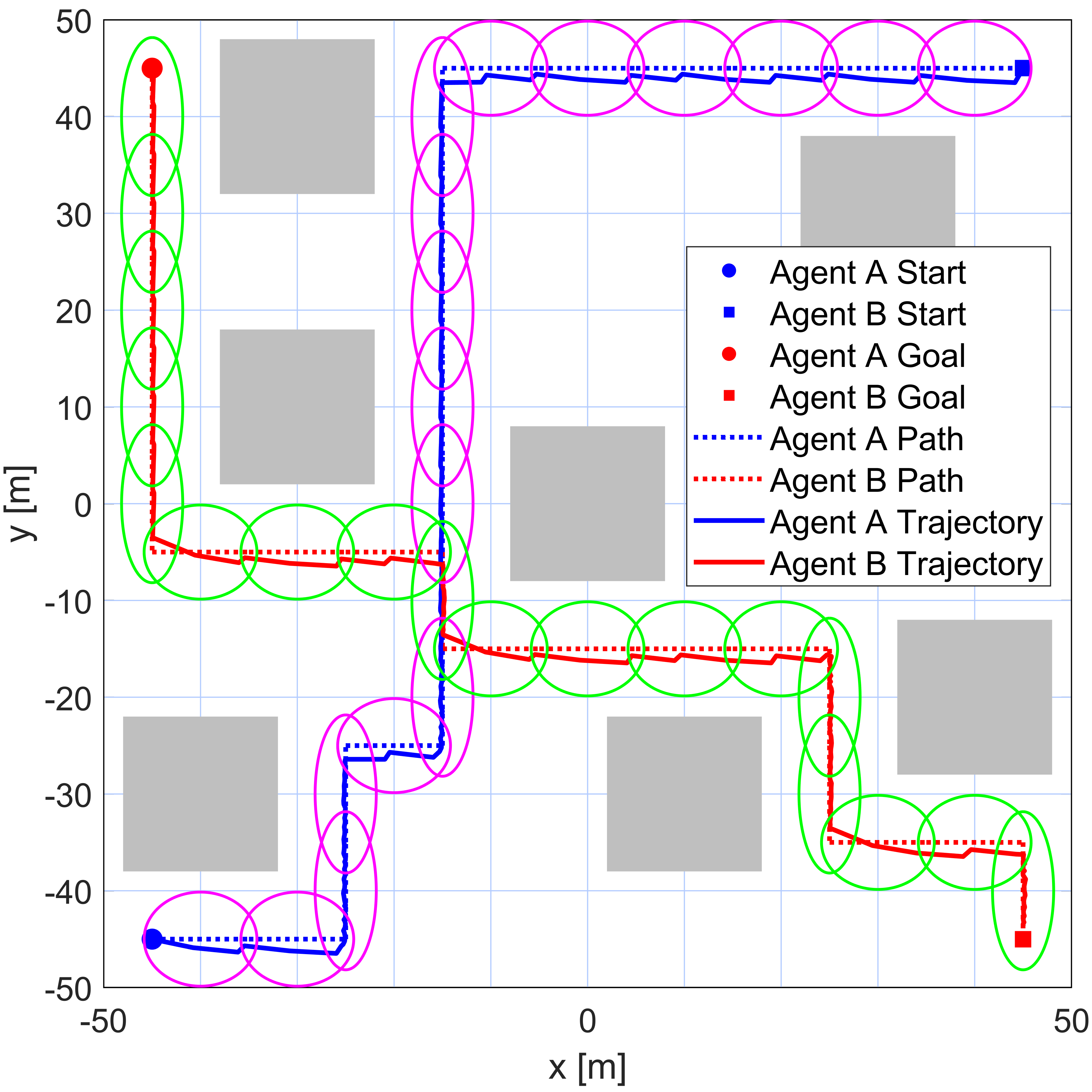}}
    \vskip\baselineskip
    \subcaptionbox{Baseline LQR--RRT executed trajectories without certified safety guarantees.\label{fig:MA_path_exec_lqr}}
        {\includegraphics[width=0.9\linewidth]{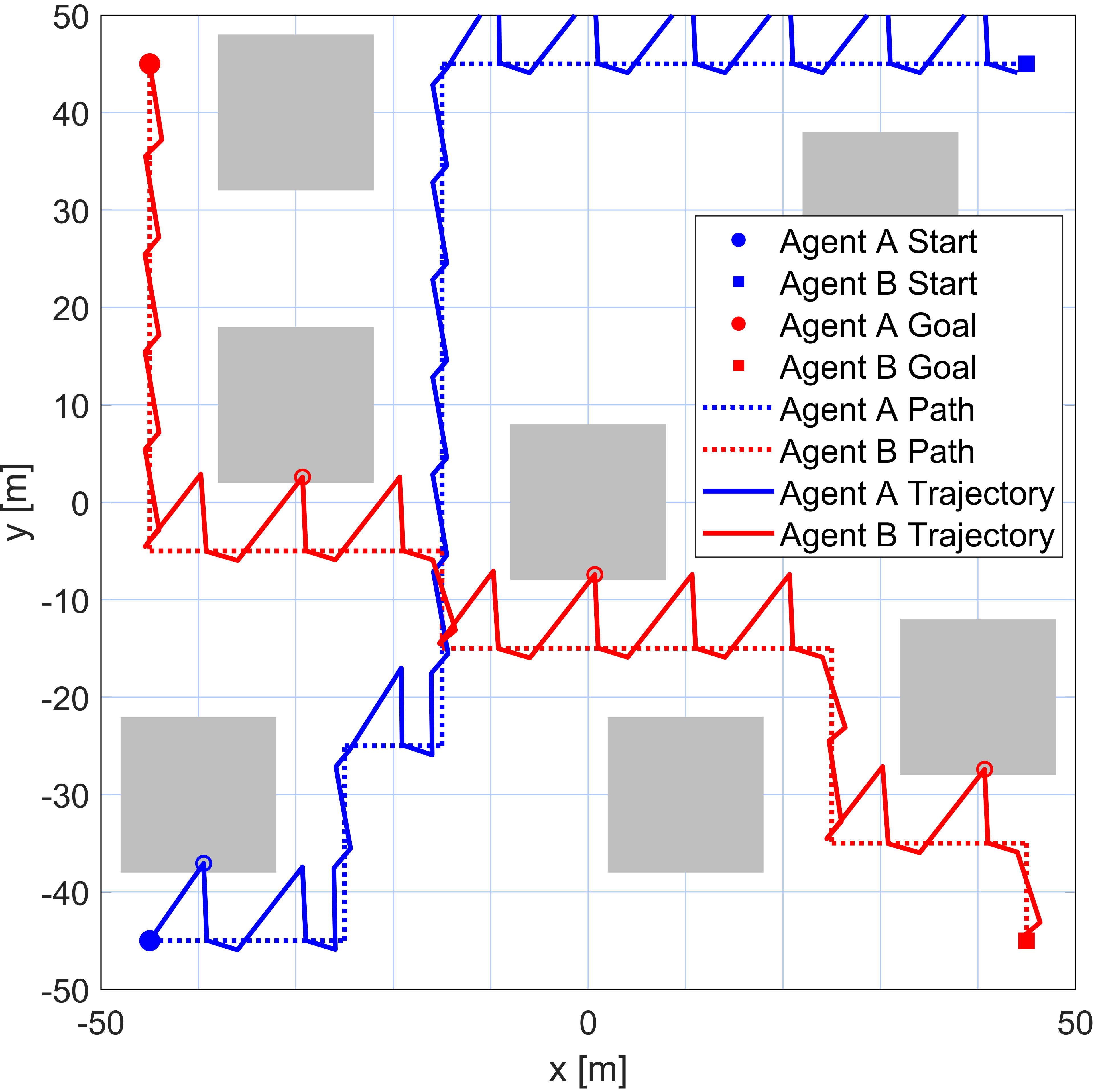}}
    \caption{Two-agent spacecraft rendezvous: comparison of executed closed-loop trajectories. 
             (a) The proposed Safe--RRT ensures trajectories remain within workspace constraints and respect space--time reservations. 
             (b) The baseline LQR--RRT produces feasible paths but lacks certified safety guarantees.}
    \label{fig:MA_paths_exec_comparison}
\end{figure}

Figure~\ref{fig:MA_tree} shows both exploration trees at the moment all
goals are reached. Figure~\ref{fig:MA_ell} overlays every certificate
computed along the final paths; the blue (Agent~A) and red (Agent~B)
ellipsoids are disjoint in space–time, so the agents remain at least
one grid cell apart for all layers. Because each edge has its own gain
that makes the local rectangle invariant, the executed trajectories are
dynamics-feasible and satisfy the state constraints. The clean
paths in Figure~\ref{fig:MA_paths_plan_sub} highlight the coordinated solution
through clutter.

We also implemented the data-driven linear–quadratic regulator
(LQR) method of~\cite{van2020noisy}, using the same path from SAFE--MA--RRT but replacing the per-edge affine controllers with LQR gains computed from data. The controller minimized the standard infinite-horizon quadratic cost, with the state weighting matrix $Q=\mathrm{diag}(1,1,0.1,0.1)$ chosen to penalize deviations in position more strongly than velocity, and 
the input weighting $R=10I$ used to limit aggressive control actions. 
While this LQR design successfully tracked the waypoints, the executed trajectories 
(Fig.~\ref{fig:MA_paths_exec_comparison}b) violated safety constraints: 
Agent~A showed violations in $7$ of $45$ layers ($15.6\%$), and Agent~B in $3$ of $45$ layers ($6.7\%$). 
By contrast, the proposed SAFE--MA--RRT (Fig.~\ref{fig:MA_paths_exec_comparison}a) maintained $0\%$ violations 
for both agents, as also summarized in Table~\ref{tab:violations}.

\begin{table}[h!]
    \centering
    \caption{Safety Violations in Executed Trajectories}
    \label{tab:violations}
    \begin{tabular}{lcc}
        \toprule
        & \textbf{SAFE--MA--RRT} & \textbf{LQR--RRT} \\
        \midrule
        Agent A & $0\%$ & $15.6\%$ \\
        Agent B & $0\%$ & $6.7\%$ \\
        \bottomrule
    \end{tabular}
\end{table}

\section{Conclusion}
A data-driven safe motion planning framework was presented for multi-agent systems with linear dynamics, enabling agents to operate in a shared environment while inter-agent safety was enforced through temporal coordination. Ellipsoidal invariant sets were constructed directly from system trajectory data, removing the need for explicit model identification. A global space–time reservation table was employed to prevent collisions, and each local motion step was certified through semidefinite programs to satisfy state and input constraints. Agent trees were expanded in an interleaved fashion, allowing synchronized yet simultaneous planning across all agents. Simulation results demonstrated that the proposed approach yields coordinated, safe, and dynamically feasible paths for all agents. In future work, the framework will be extended to account for noise by incorporating probabilistic safety guarantees, and to support multi-agent systems with nonlinear dynamics.

\ifCLASSOPTIONcaptionsoff
  \newpage
\fi

\bibliographystyle{IEEEtran}
\bibliography{Refs}

\vfill

\end{document}